\documentclass[11pt]{article}
\usepackage{fullpage}
\usepackage{graphicx}
\usepackage{amsmath,amssymb,amsthm,amsfonts,dsfont,mathtools,latexsym}
\usepackage{relsize}
\usepackage[margin=1cm]{caption}
\usepackage{wrapfig}
\usepackage{frame,color}
\usepackage{environ,subfig}
\usepackage{hyperref}
\usepackage{xspace}
\usepackage{framed}
\usepackage{comment}
\usepackage{indentfirst}
\usepackage{enumerate,braket}

\theoremstyle{plain}
\newtheorem{theorem}{Theorem}[section]
\newtheorem{lemma}[theorem]{Lemma}

\newtheorem{corollary}[theorem]{Corollary}
\newtheorem{problem}{Problem}
\newtheorem{proposition}[theorem]{Proposition}
\theoremstyle{definition}
\newtheorem{definition}{Definition}[section]
\newtheorem{remark}{Remark}
\newtheorem{example}{Example}

\newenvironment{proofof}[1]{\begin{proof}[Proof of #1]}{\end{proof}}

\newcommand{\alg}{\mathcal{A}}
\newcommand{\apxc}{\mathcal{C}}
\newcommand{\poly}{\operatorname{poly}}
\newcommand{\polylog}{\operatorname{polylog}}
\newcommand{\polyloglog}{\operatorname{polyloglog}}

\DeclareMathOperator{\arank}{rank}
\DeclareMathOperator{\rank}{rank}

\newcommand{\BQP}{\mathsf{BQP}}

\newcommand{\AM}{\mathsf{AM}}
\newcommand{\AMcc}{\mathsf{AM}^{\cc}}

\newcommand{\PTIME}{\mathsf{P}}

\newcommand{\PH}{\mathsf{PH}}
\newcommand{\PHcc}{\mathsf{PH}^{\cc}}

\newcommand{\UPP}{\mathsf{UPP}}
 
\newcommand{\cc}{\mathsf{cc}}
\newcommand{\eps}{\epsilon}

\newcommand{\SATPAIR}{\textsf{-Satisfying-Pair}}
\renewcommand{\epsilon}{\varepsilon}

\newcommand{\MaxIP}{\textsf{Max-IP}}

\newcommand{\MAX}{\textsf{Max}}

\newcommand{\LCS}{\textsf{LCS}}
\newcommand{\Edit}{\textsf{Edit-Distance}}
\newcommand{\DistLCS}{\textsf{LCS}^{\cc}}
\newcommand{\DistEdit}{\textsf{Edit-Dist}^{\cc}}

\newcommand{\SAT}{\textsf{SAT}}
\newcommand{\SETH}{\textsf{SETH}}

\newcommand{\EAMcc}{\AMcc_{\sf eff}}

\newcommand{\SYM}{\textsf{SYM}}
\newcommand{\OV}{\textsf{OV}}

\newcommand{\domainX}{\mathcal{X}}
\newcommand{\domainY}{\mathcal{Y}}

\newcommand{\countkOV}{\textsf{\#$k$-OV}}
\newcommand{\countOV}{\textsf{\#OV}}

\newcommand{\MGapIP}{\textsf{Gap-Inner-Product}}

\newcommand{\AND}{\textsf{AND}}
\newcommand{\OR}{\textsf{OR}}

\newcommand{\SETDISJ}{\textsc{Set-Disjointness}}
\newcommand{\ED}{\textsc{Element Distinctness}}
\newcommand{\FormulaEval}{\textsc{Formula Evaluation}}
\newcommand{\sparsecountOV}{\textsf{\#Sparse-OV}}
\newenvironment{reminder}[1]{\bigskip
	\noindent {\bf Reminder of #1  }\em}{\smallskip}

\title{Classical Algorithms from Quantum and Arthur-Merlin Communication Protocols}
\author{Lijie Chen \\Massachusetts Institute of Technology \\ \texttt{lijieche@mit.edu} \and Ruosong Wang\\Carnegie Mellon University \\ \texttt{ruosongw@andrew.cmu.edu}}
\date{}
\begin{document}
	\maketitle
	
	\begin{abstract}
		In recent years, the polynomial method from circuit complexity has been applied to several fundamental problems and obtains the state-of-the-art running times (e.g., R. Williams's $n^3 / 2^{\Omega(\sqrt{\log n})}$ time algorithm for APSP). As observed in [Alman and Williams, STOC 2017], almost all applications of the polynomial method in algorithm design ultimately rely on certain (probabilistic) low-rank decompositions of the computation matrices corresponding to key subroutines. They suggest that making use of low-rank decompositions directly could lead to more powerful algorithms, as the polynomial method is just one way to derive such a decomposition.
		
		Inspired by their observation, in this paper, we study another way of \emph{systematically constructing low-rank decompositions of matrices} which could be used by algorithms\textemdash\emph{communication protocols}. Since their introduction, it is known that various types of communication protocols lead to certain low-rank decompositions (e.g., $\PTIME$ protocols/rank, $\BQP$ protocols/approximate rank). These are usually interpreted as approaches for proving communication lower bounds, while in this work we explore the other direction. 
		
		We have the following two generic algorithmic applications of communication protocols:
		
		\begin{itemize}
			\item \textbf{Quantum Communication Protocols and Deterministic Approximate Counting.} Our first connection is that a fast $\BQP$ communication protocol for a function $f$ implies a fast deterministic additive approximate counting algorithm for a related pair counting problem. Applying known $\BQP$ communication protocols for $\SETDISJ$, $\ED$ and $\FormulaEval$, we get fast deterministic additive approximate counting algorithms for Count-OV (\countOV), Sparse Count-OV and Formula of $\SYM$ circuits. In particular, our approximate counting algorithm for $\countOV$ runs in near-linear time for all dimensions $d = o(\log^2 n)$. Previously, even no truly-subquadratic time algorithm was known for $d = \omega(\log n)$.
			
			\item \textbf{Arthur-Merlin Communication Protocols and Faster $\textsf{Satisfying-Pair}$ Algorithms.} Our second connection is that a fast $\AMcc$ protocol for a function $f$ implies a faster-than-bruteforce algorithm for $f\SATPAIR$. Using the classical Goldwasser-Sisper $\AM$ protocols for approximating set size, we obtain a new algorithm for approximate $\MaxIP_{n,c\log n}$ in time $n^{2 - 1/O(\log c)}$, matching the state-of-the-art algorithms in [Chen, CCC 2018].
		\end{itemize}
	
		We also apply our second connection to shed some light on long-standing open problems in communication complexity. We show that if the Longest Common Subsequence ($\LCS$) problem admits a fast (computationally efficient) $\AMcc$ protocol ($\polylog(n)$ complexity), then polynomial-size Formula-$\SAT$ admits a $2^{n - n^{1-\delta}}$ time algorithm for any constant $\delta > 0$, which is conjectured to be unlikely by a recent work [Abboud and Bringmann, ICALP 2018]. The same holds even for a fast (computationally efficient) $\PHcc$ protocol.
	\end{abstract}
	
	\thispagestyle{empty}
	\newpage
\section{Introduction}

Recent works have shown that the polynomial method, a classical technique for proving circuit lower bounds~\cite{Raz87,Smo87}, can be useful in designing efficient algorithms~\cite{Wil14a,Wil14b,AWY15,AW15,ACW16,LPTWY17,Alman19LightBulb}.

At a very high level, these algorithms proceed as follows: (1) identify a key subroutine of the core algorithm which has a certain low-degree polynomial representation; (2) replace that subroutine by the corresponding polynomials, and reduce the whole problem to a certain \emph{batched evaluation problem of sparse polynomials}; (3) embed that polynomial evaluation problem to \emph{multiplication of two low-rank (rectangular) matrices}, and apply the fast rectangular matrix multiplication algorithm~\cite{coppersmith1982rapid}.

As~\cite{AW17} point out. In term of step (3), these algorithms are ultimately making use of the fact that the corresponding matrices of some circuits or subroutines have low \emph{probabilistic rank}.\footnote{See Appendix~\ref{app:prob-low-rank-OV} for an illustration with the $n^{2 - 1/O(\log c)}$ time algorithm for $\OV_{n,c\log n}$ in~\cite{AWY15}.}~\cite{AW17} suggest that the \emph{probabilistic rank}, or various low-rank decompositions of matrices in general\footnote{A low probabilistic rank implies a probabilistic low-rank decomposition of the matrix.}, could be more powerful than the polynomial method, and lead to more efficient algorithms, as the polynomial method is just one way to construct them.

It has been noted for a long time that communication protocols are closely related to various notions of rank of matrices. To list a few: deterministic communication complexity is lower bounded by the logarithm of the \emph{rank} of the matrix~\cite{MS82}; quantum communication complexity is lower bounded by the logarithm of the \emph{approximate rank} of the matrix~\cite{ASTVW03,BCW98-quantum_communication}; $\UPP$ communication complexity is equivalent to the logarithm of the \emph{sign-rank} of the matrix~\cite{PS86}.

These connections are introduced (and usually interpreted) as methods for proving communication complexity lower bounds (see, e.g. the survey by Lee and Shraibman~\cite{LS09}), but they can also be interpreted in the other direction, as a way to \emph{systematically construct low-rank decompositions of matrices}. 

In this paper, we explore the connection between different types of communication protocols and low-rank decompositions of matrices and establish several applications in algorithm design. 
For all these connections, we start with an efficient communication protocol for a problem $F$, which implies an efficiently constructible low-rank decomposition of the corresponding communication matrix of $F$, from which we can obtain fast algorithms.

In fact, in our applications of quantum communication protocols, we also consider $k$-party protocols, and our algorithms rely on the approximate low-rank decomposition of the tensor of the corresponding communication problem. To the best of our knowledge, this is the first time that \emph{approximate tensor rank} is used in algorithm design (approximate rank has been used before, see e.g.~\cite{alon2009perturbed,barak2011rank,alon2013approximate,alon2014cover} and the corresponding related works section).\footnote{We remark that a concurrent work~\cite{yu2018optimal} makes algorithmic use of \emph{non-negative tensor approximate rank} to construct an optimal data structure for the succinct rank problem.} 

\subsection{Quantum Communication Protocols and Deterministic Approximate Counting}

Our first result is a generic connection between quantum communication protocols and deterministic approximate counting algorithms.

\begin{theorem}\label{thm:apx_count-two-party}(Informal)
	Let $\domainX, \domainY$ be finite sets and $f :  \domainX \times \domainY \to \{0, 1\}$ be a Boolean function. Suppose $f$ has a quantum communication protocol $\mathcal{P}$\footnote{We need some technical condition on $\mathcal{P}$, see Corollary~\ref{cor:apx_count} for details.} with complexity $C(\mathcal{P})$ and error $\varepsilon$. Then there is a classical deterministic algorithm $\apxc$ that receives $A \subseteq \domainX, B \subseteq \domainY$ as input, and outputs a number $E$ such that
	\begin{align*}
	\left|
	\sum_{(x,y) \in A \times B} f(x,y) - E  
	\right| \le \epsilon \cdot |A| \cdot |B|.
	\end{align*}
	Furthermore, $\apxc$ runs in $(|A| + |B|) \cdot 2^{O(C(\mathcal{P}))}$ time.
\end{theorem}

We remark here that there is a simple randomized algorithm running in sub-linear time via random-sampling. Thus the above algorithm is indeed a derandomization of that randomized algorithm.

The above theorem can also be easily generalized to the (number-in-hand) $k$-party case. See Section \ref{sec:def_quantum_comm} for the definition of the multiparty quantum communication model. 

\begin{theorem}\label{thm:apx_count-general}(Informal)
	Let $\domainX_1, \domainX_2, \ldots, \domainX_k$ be finite sets and $f :  \domainX_1, \domainX_2, \ldots, \domainX_k \to \{0, 1\}$ be a Boolean function.
	Suppose $f$ has a $k$-party quantum communication protocol $\mathcal{P}$ with complexity $C(\mathcal{P})$ and error $\epsilon$. Then there is a classical deterministic algorithm $\apxc$ that receives $X_1 \subseteq \domainX_1, X_2 \subseteq \domainX_2, \ldots, X_k \subseteq \domainX_k$ as input, and outputs a number $E$ such that
	\begin{align*}
	\left|
	\sum_{x_1 \in X_1, x_2 \in X_2, \ldots, x_k \in X_k} f(x_1, x_2, \ldots, x_k) - E
	\right|
	\le \epsilon \cdot \prod_{i=1}^{k} |X_i| .
	\end{align*}
	Furthermore, $\apxc$ runs in $(|X_1| + |X_2| + \ldots + |X_k|) \cdot 2^{O(C(\mathcal{P}))}$ time.
\end{theorem}

\newcommand{\sketch}{\textsf{sk}}

\paragraph*{Sketching Algorithms.} In fact, Theorem~\ref{thm:apx_count-general} implies a stronger \emph{sketching algorithm}. Given subsets $X_1,X_2,\dotsc,X_k$, the algorithm first computes a $w =2^{O(C(\mathcal{P}))}$ size sketch $\sketch_i$ from each $X_i$ in $O(|X_i| \cdot w)$ time deterministically, and the number $E$ can be computed from these $\sketch_i$'s in $O(k \cdot w)$ time.

The sketch computed by the algorithm is in fact a vector in $\mathbb{R}^{w}$, and it satisfies a nice additive property. That is, the sketch of $X_1 \sqcup X_2$ (union as a multi-set) is simply $\sketch(X_1) + \sketch(X_2)$.


Applying existing quantum communication protocols, we obtain several applications of Theorem~\ref{thm:apx_count-two-party} and Theorem~\ref{thm:apx_count-general}.

\subsubsection{\SETDISJ\ and Approximate $\textsf{\#OV}$ and $\countkOV$}

We first consider the famous \SETDISJ\ problem (Alice and Bob get two vectors $u$ and $v$ in $\{0, 1\}^d$ correspondingly, and want to determine whether $\langle u , v \rangle = 0$), which has an efficient quantum communication protocol~\cite{AA05} with communication complexity $O(\sqrt{d})$.

The corresponding count problem for \SETDISJ\ is the counting version of the Orthogonal Vectors problem ($\OV$), denoted as $\textsf{\#OV}_{n,d}$. In this problem, we are given two sets of $n$ vectors $S, T\subseteq \{0, 1\}^{d}$, and the goal is to count the number of pairs $u \in S, v \in T$ such that $\langle u, v \rangle = 0$.

Applying the quantum communication protocol for $\SETDISJ$ and Theorem~\ref{thm:apx_count-general}, we immediately get an algorithm for the approximate version of $\countOV$.

\begin{theorem}\label{thm:apx-OV}
	For any $d$ and any constant $\epsilon > 0$, $\textsf{\#OV}_{n,d}$ can be approximated deterministically with additive error $\epsilon \cdot n^2$ in $n \cdot 2^{O(\sqrt{d})}$ time. In particular, it runs in $n^{1 + o(1)}$ time when $d = o(\log^2 n)$.
\end{theorem}

\paragraph*{Comparison with~\cite{CW16}.} \cite{CW16} gives a \emph{deterministic exact counting} algorithm for $\countOV_{n,c\log n}$, which runs in $n^{2 - O(1/\log c)}$ time. Note that their running time is $n^{2 - o(1)}$ when $d = \omega(\log n)$, while our algorithm only achieves an additive approximation, but runs in \emph{near-linear} time for all $d = o(\log^2 n)$.

Another closely related problem, \textsc{Counting Partial Match}, is the problem that given $n$ query strings from $\{0, 1, \star\}^{d}$ ($\star$ is a ``don't care'') and $n$ strings from $\{0, 1\}^{d}$, and the goal is to count the number of matching string and query pairs.

Using known reductions between \textsc{Partial Match} and \textsf{OV} (see, e.g., Section 2 in \cite{AWY15}), together with the approximate counting algorithm for $\countOV$, we can also solve \textsc{Counting Partial Match} approximately in the same running time.

The approximate counting algorithm for $\countOV$ can be easily generalized to solve \textsf{\#$k$-OV}, which is the problem that given $k$ sets of $n$ vectors $X_1, X_2, \ldots, X_k \subseteq \{0, 1\}^{d}$, and count the number of $k$-tuples $u_1 \in X_1, u_2 \in X_2, \ldots, u_k \in X_k$ such that $\langle u_1, u_2, \ldots, u_k \rangle = 0$.\footnote{the generalized inner product of $k$ vectors, is defined as $\langle u_1, u_2, \ldots, u_k \rangle = \sum_{i=1}^{d} \prod_{j=1}^k (u_j)_i$.}

Applying Theorem~\ref{thm:apx_count-general} and observe that the $2$-party \SETDISJ\ protocol in~\cite{AA05} can be easily generalized to solve the $k$-party case (in $k$-party $\SETDISJ$, there are $k$ players getting $u_1,u_2,\dotsc,u_k$ respectively, and they want to determine whether $\langle u_1, u_2, \ldots, u_k \rangle = 0$), we obtain the following approximate counting algorithm for \countkOV.

\begin{theorem}\label{thm:count_k_ov}
	For any integers $k$, $d$ and any constant $\epsilon > 0$, $\countkOV_{n,d}$ can be approximated deterministically with additive error $\epsilon \cdot n^k$ in $n \cdot 2^{O(k\sqrt{d})}$ time. In particular, it runs in $n^{1 + o(1)}$ time when $k$ is a constant and $d = o(\log^2 n)$.
\end{theorem}

\begin{remark}
We remark that similar algorithms with slightly worse running time ($n \cdot d^{O(\sqrt{d} )}$ time for additive approximation to $\countOV_{n,d}$) can also be derived using the polynomial method, see Appendix~\ref{app:algo-from-approx-poly} for details. However, we think our new algorithms via quantum communication protocols have the following extra benefits: (1) our algorithm is slightly faster, with a running time of $n \cdot 2^{O(\sqrt{d})}$; (2) our algorithm is derived via a general connection. Once the connection is set up, the algorithm follows in an elegant and black-box way. We hope this general connection could stimulate more applications of quantum communication protocols.
\end{remark}

\subsubsection{Sparse $\SETDISJ$ and Approximate Sparse $\textsf{\#OV}$}

Next we consider a sparse version of $\SETDISJ$, in which Alice and Bob get two sparse vectors $u,v \in \{0,1\}^m_{\le d}$\footnote{We use $\{0, 1\}^m_{\le d}$ to denote all Boolean vectors of length $m$ with at most $d$ ones.}, and want to decide whether $\langle u, v\rangle = 0$.

Using the famous quantum-walk algorithm for $\ED$~\cite{ambainis2007quantum}, there is an $O(d^{2/3} \log m)$ communication protocol for sparse $\SETDISJ$, which is much better than the $O(\sqrt{m})$ protocol for $\SETDISJ$ when $m \gg d$.

Applying this protocol and Theorem~\ref{thm:apx_count-two-party}, we can give an algorithm for a sparse version of $\countOV$, denoted as $\sparsecountOV_{n,m,d}$, in which we are given sets $A,B \subseteq \{0, 1\}^{m}_{\le d}$ of $n$ vectors, and the goal is to count the number of distinct $(a,b) \in A \times B$ such that $\langle a,b \rangle = 0$. Formally, we have:
\begin{theorem}\label{thm:count_sparse_ov}
	For integers $n,m,d$ and any constant $\epsilon > 0$, $\sparsecountOV_{n,m,d}$ can be approximated deterministically with additive error $\epsilon \cdot n^2$ in 
	\[
	n \cdot 2^{O(d^{2/3} \log(m))}
	\]
	time. In particular, when $m = \poly(d)$ and $d = o\left( \left(\frac{\log n}{\log 
		\log n} \right)^{1.5} \right)$, it runs in $n^{1 + o(1)}$ time.
\end{theorem} 

We remark that it is possible to improve Theorem~\ref{thm:count_sparse_ov} via the polynomial method (see Appendix~\ref{app:algo-from-approx-poly} for details). Again, we emphasize that our focus here is to provide direct applications of our general framework, with the hope that it could stimulate more applications of quantum communication protocols in the classical settings.

\subsubsection{Approximate Counting for $\textsf{Formula} \circ \SYM$ Circuits}

Finally, we apply our algorithm to approximately count solutions (i.e., satisfying assignments) to a class of circuits, for which no non-trivial algorithms were previously known.

A $\textsf{Formula} \circ \SYM$ circuit of size $m$ is a formula with $\{ \textsf{AND}, \textsf{OR},\textsf{NOT} \}$ basis on $m$ $\SYM$ gates\footnote{A $\SYM$ gate is a gate whose output only depends on the number of ones in the input.} at the bottom. Using the quantum query algorithm for $\FormulaEval$~\cite{ambainis2010any} and the split-and-list technique, we obtain the following deterministic approximate counting algorithm for $\textsf{Formula} \circ \SYM$ circuits:

\begin{theorem}\label{thm:count_formula_sym}
	For any constant $\epsilon > 0$, the number of solutions to a $\textsf{Formula} \circ \SYM$ circuit of size $m$ can be approximated deterministically within $\epsilon \cdot 2^n$ additive error in
	\[
	2^{O(n^{1/2} m^{1/4 + o(1)} \sqrt{\log n + \log m})}
	\]
	time. In particular, when $m = n^{2-\delta}$ for some $\delta > 0$, the running time is $2^{o(n)}$.
\end{theorem}

Previously, even no non-trivial deterministic approximate counting algorithms for $\AND \circ \SYM$ circuits were known. A recent line of works~\cite{GOWZ10,HKM12,ST17a}, culminating in~\cite{DST18}, construct a PRG for $\AND_m \circ \textsf{THR}$ circuits with seed length $\poly(\log m,\delta^{-1}) \cdot \log n$, using which one can obtain a quasi-polynomial time deterministic approximate counting algorithm for polynomial size $\AND \circ \textsf{THR}$ circuits. However, their PRG constructions rely on the fact that the solution set of an $\AND_m \circ \textsf{THR}$ circuit is a \emph{polytope}, while the solution set of an $\AND \circ \SYM$ circuit may not have such a nice geometric structure.

In fact, the only property we need for $\SYM$ gates is that they admit an efficient \emph{classical} $k$-party communication protocol when the inputs are divided to $k$ players (each player sends the contribution of her part). Our algorithm actually works for the following more general problem. 

\begin{problem}\label{pro:count_sym}
	Given $k$ sets of $n$ vectors $X_1, X_2, \ldots, X_k \subseteq \{0,\dotsc,r\}^{d}$ and $d$ functions $f_1, f_2, \ldots, f_d$ where each $f_i$ is from $[r]^k$ to $\{0, 1\}$,
	and a Boolean formula $\mathcal{F} : \{0, 1\}^d \to \{0, 1\}$ of $O(1)$ fan-in.
	Count the number of $k$-tuples $u_1 \in X_1, u_2 \in X_2, \ldots, u_k \in X_k$ such that
	$$
	\mathcal{F}(f_1(u_{1, 1}, u_{2, 1}, \ldots, u_{k, 1}),  f_2(u_{1, 2}, u_{2, 2}, \ldots, u_{k, 2}), \ldots, f_d(u_{1, d}, u_{2, d}, \ldots, u_{k, d})) = 1.
	$$
\end{problem}

\begin{theorem}\label{thm:count_sym_split}
	For any constant $\epsilon > 0$, the above problem can be solved deterministically in $n \cdot 2^{O(d^{1/2+o(1)} \cdot k (\log d + \log r) )}$ time, within $\varepsilon \cdot n^k$ additive error. 
\end{theorem}

\subsection{Arthur-Merlin Communication Protocols and a New Approximate $\MaxIP$ Algorithm}

Our second connection is an algorithmic application of $\AMcc$ protocols. We first define $\AMcc$ protocols formally.

\begin{definition}
	An Arthur-Merlin communication protocol ($\AMcc$) $\Pi$ for a partial function $F : \domainX \times \domainY \to \{0,1,\bot\}$\footnote{$F(x,y) = \bot$ means $F(x,y)$ is undefined.} proceeds as follows:
	
	\begin{itemize}
		\item Alice holds input $x \in \domainX$ and Bob holds input $y \in \domainY$.
		\item Alice and Bob toss some public coins jointly and send the random string $r \in \{0,1\}^*$ to Merlin ($r$ is called the random challenge).
		\item Based on $x$, $y$ and the random challenge $r$, Merlin sends Alice and Bob a proof $\psi$, and Alice and Bob decide to accept or not independently and deterministically. We require the following conditions:
		
		\begin{itemize}
			\item If $F(x,y) = 1$, with probability $1 - \epsilon$ over the random challenge $r$, there is a proof $\psi$ from Merlin such that Alice and Bob both accept.
			\item If $F(x,y) = 0$, with probability $1 - \epsilon$ over the random challenge $r$, there is no proof $\psi$ from Merlin such that Alice and Bob both accept.
		\end{itemize}
	\end{itemize} 
	
	We call the parameter $\epsilon$ the error of the protocol $\Pi$. Moreover, we say the protocol is \emph{computationally efficient} if Alice and Bob's behavior can be computed in polynomial-time w.r.t. their input lengths.
\end{definition}

We show that for any function $F$, a low-complexity and computationally efficient $\AMcc$ protocol implies a faster algorithm for the corresponding $F\SATPAIR$ problem (defined below).

For a partial function $F : \domainX \times \domainY \to \{0,1,\bot\}$, where $\domainX$ and $\domainY$ are two sets, we define $F\SATPAIR_n$ as the problem that given two sets $A \subseteq \domainX$ and $B\subseteq \domainY$ of size $n$, distinguish between the following two cases: (1) There is an $(x,y) \in A \times B$ such that $F(x,y) = 1$. (2) For all $(x,y) \in A \times B$, $F(x,y) = 0$.

\begin{theorem}[Algorithms from $\AMcc$ protocols] \label{thm:alg_amcc}
	Let $F : \domainX \times \domainY \to \{0,1,\perp\}$ be a partial function. Suppose there is a computationally efficient $\AMcc$ protocol for $F$ with communication complexity $T$ and error $\epsilon$. Then for $n$ such that $2^{T} \le (\sqrt{\varepsilon}n)^{0.1}$, there is an $O\left( \epsilon n^2 \cdot \polylog(n)  + n \cdot 2^T\right)$ time randomized algorithm for $F\SATPAIR_n$.
\end{theorem}

\subsubsection{A New Algorithm for Approximate $\MaxIP$}

The first application of Theorem~\ref{thm:alg_amcc} is a new algorithm for approximate Maximum Inner Product. We use $\MaxIP_{n,d}$ to denote the problem that given sets $A, B \subseteq \{0,1\}^d$ with size $n$, compute $\MAX(A,B) := \max_{(a,b) \in A \times B} \langle a \cdot b \rangle$.

To phrase this as an $F\SATPAIR$ problem, we first define the following gap inner product problem.

\begin{definition}[Multiplicative-Gap Inner Product]\label{defi:MGapIP}
	Consider the following problem, denoted as $\MGapIP_d$, Alice and Bob hold strings $x,y \in \{0,1\}^d$ respectively, and they are given an integer $\tau$. They want to distinguish between the following two cases: (Yes) $x \cdot y \ge 2\tau$; (No) $x \cdot y \le \tau$.
\end{definition}

Adapting the classical Goldwasser-Sisper $\AM$ protocol for approximating set size~\cite{GS89}, we can derive an efficient $\AMcc$ protocol for $\MGapIP_d$.

\begin{lemma}[$\AMcc$ protocol for $\MGapIP_d$]\label{lm:AMcc-MGapIP}
	There is an $\AMcc$ protocol for $\MGapIP_d$ with error $\epsilon$ and communication complexity
	\[
	\log \binom{d}{\le O(\log \epsilon^{-1})}.\footnote{$\binom{n}{\le m}$ denotes $\sum_{i=0}^{m} \binom{n}{i}$.}
	\]
\end{lemma}

Applying Theorem~\ref{thm:alg_amcc}, the following algorithm for approximating $\MaxIP$ follows directly, matching the previous best algorithm in~\cite{Che18}.

\begin{corollary}\label{cor:algo-apx-max-ip}
	There is an algorithm for computing a $2$-approximation to $\MaxIP_{n,c \log n}$, which runs in $n^{2 - 1/O(\log c)}$ time.
\end{corollary}
\begin{remark}
	The constant $2$ in Corollary \ref{cor:algo-apx-max-ip} can be replaced by any other constant $\kappa > 1$.
\end{remark}

We remark here that a direct application of the Goldwasser-Sisper protocol and parallel repetition leads to a communication protocol with communication complexity $O(\log d \log \epsilon^{-1})$, which is slightly worse than Lemma~\ref{lm:AMcc-MGapIP}. In particular, such a protocol only gives an algorithm with running time $n^{2 - 1 /O(\log d)}$, which is worse than $n^{2 - 1/O(\log c)}$ when $c \ll d = c \log n$. In order to get the improved complexity in Lemma~\ref{lm:AMcc-MGapIP}, we make use of a clever sampling scheme using Poisson distributions, see Section~\ref{sec:amcc-max-ip} for details.

\subsubsection{Evidences that Longest Common Subsequence and Edit Distance do not Have Fast $\AMcc$ Protocols} 

It has been a long-standing open problem in communication complexity to prove an $\omega(\log n)$ $\AMcc$ lower bound for any explicit function~\cite{BFS86,GoosPW16,GPW18}\textemdash it is consistent with our current knowledge that all known natural communication problems have $O(\log n)$ $\AMcc$ protocols. 

We consider two natural communication problems here, $\DistLCS_{d}$ and $\DistEdit_d$, in which Alice and Bob hold strings $x,y \in \{0,1\}^d$ respectively, and are given an integer $\tau$. Their goal is to decide whether $\LCS(x,y) \ge \tau$ ($\Edit(x,y) \ge \tau$).

Our Theorem~\ref{thm:alg_amcc} shows that if $\DistLCS$ or $\DistEdit$ admit low-complexity and computationally efficient $\AMcc$ protocols, it would imply non-trivial algorithms for the corresponding $F\SATPAIR$ problem. By a known reduction in~\cite{AHVW16}, that would, in turn, implies non-trivial algorithms for Formula-$\SAT$\footnote{Formula-$\SAT$ is the problem that deciding whether a given formula is satisfiable.}\textemdash much faster than the current state-of-the-art~\cite{Tal15}! Therefore, at least for these two problems, constructing low-complexity $\AMcc$ protocol could be hard, which may also be viewed as an evidence that they do not have efficient $\AMcc$ protocols.

\begin{theorem}\label{thm:consequence-LCS}
	If $\DistLCS_d$ admits computationally efficient $\AMcc$ protocols with complexity $\polylog(d)$, then Formula-$\SAT$ of polynomial-size formulas admits an $2^{n - n^{1 - \delta}}$ time algorithm for any constant $\delta > 0$. The same holds for $\DistEdit$ in place of $\DistLCS$.
\end{theorem}

The state-of-the-art algorithm for Formula-$\SAT$ runs in $o(2^n)$ time only when the formula size is smaller than $n^3$~\cite{Tal15}. It is even purposed as a hypothesis that no $2^{n} / n^{\omega(1)}$ time algorithm exists for $n^{3+\Omega(1)}$-size Formula-$\SAT$ in~\cite{AB18}. Therefore, our results imply that if $\DistLCS$ or $\DistEdit$ admits fast (computationally efficient) $\AMcc$ protocols, then that would refute the hypothesis in~\cite{AB18}:

\begin{corollary}\label{cor:condtional-lowb-for-AM}
	Under the following hypothesis\footnote{which is much weaker than the hypothesis in~\cite{AB18}}, $\DistLCS_d$ and $\DistEdit_d$ do not admit computationally efficient $\AMcc$ protocols with complexity $\polylog(d)$:
	
	\begin{itemize}
		\item There is a constant $\delta > 0$ such that Formula-$\SAT$ of polynomial-size formulas requires $2^{n - n^{1-\delta}}$ time.
	\end{itemize}
	
\end{corollary}

In fact, in Appendix~\ref{app:PH-protocols}, we show that the above corollary can be generalized to hold for computationally efficient $\PHcc$ protocols (see Appendix~\ref{app:PH-protocols} for a formal definition). Formally, we have:

\begin{theorem}\label{theo:condtional-lowb-for-PH}
	Under the same hypothesis as in Corollary~\ref{cor:condtional-lowb-for-AM}, $\DistLCS_d$ and $\DistEdit_d$ do not admit computationally efficient $\PHcc$ protocols with complexity $\polylog(d)$.
\end{theorem}


\subsection{Related Works}

\subsubsection{Communication Protocols and Fine-Grained Complexity}


Recently, since the breakthrough work of~\cite{ARW17}, communication protocols have been applied to \emph{fine-grained complexity}, and several tight conditional hardness results are proved for many fundamental approximate problems in $\PTIME$~\cite{ARW17,CLM18,AR18,Che18,CGLRR18Meets,CW19,CP19}.

Among these works, the most related one is~\cite{Che18}, in which the author also makes use of the $\BQP^\cc$ protocol for $\SETDISJ$ for a different purpose. In~\cite{Che18}, the $\BQP^\cc$ protocol is used to established a \emph{reduction} from $\OV$ to approximate $\{-1,1\}\text{-}\MaxIP$\footnote{a variant of $\MaxIP$ with vectors in $\{-1,1\}^d$ instead of $\{0,1\}^d$}, thereby showing the $\SETH$-hardness of approximating $\{-1,1\}\text{-}\MaxIP$.
On the other hand, in this work we use $\BQP^\cc$ protocols directly for algorithmic purposes.

\subsubsection{Other Algorithmic Applications of Approximate Rank}
Alon studies the approximate rank of the identity matrix $I_n$ in \cite{alon2009perturbed}.
It is shown that it is at least $\Omega\left(\frac{\log n}{\varepsilon^2 \log(1 / \varepsilon)}\right)$ and at most $O\left(\frac{\log n} {\varepsilon^2}\right)$.
Built upon this result, several applications in geometry, coding theory, extremal finite set theory and the study of sample spaces supporting nearly independent random variables are derived.
The lower bound also has applications in combinatorial geometry and in the study of locally correctable codes over real and complex numbers, as shown in \cite{barak2011rank}.
In \cite{alon2013approximate, alon2014cover}, several bounds on approximate rank are derived, together with applications of approximate rank in approximating Nash Equilibria, approximating densest bipartite subgraph and covering convex bodies. 

\section{Preliminaries}
\subsection{Fast Rectangular Matrix Multiplication}

Similar to previous algorithms using the polynomial method (see, e.g., \cite{Wil14b,AW15,AWY15}), our algorithms also make use of algorithms for fast rectangular matrix multiplication. 

\begin{theorem}[\cite{coppersmith1982rapid, gall2018improved}]\label{theo:fast-matrix-mult-polylog}
	There is an $N^{2} \cdot \polylog(N)$ time algorithm for multiplying two matrices $A$ and $B$ with size $N \times N^{\alpha}$ and $N^{\alpha} \times N$, where $\alpha > 0.172$.
\end{theorem}

\subsection{Random Variables and Poisson Distributions}

Throughout the paper, we use $X \simeq Y$ to mean that $X$ and $Y$ have the same distribution.
We use $X \succeq Y$ to denote stochastic dominance, i.e., $X \succeq Y$ iff for any $t \in \mathbb{R}$,  $\Pr[X \ge t] \ge \Pr[Y \ge t]$.

We use $\mathrm{Pois}(\lambda)$ to denote a Poisson distribution with parameter $\lambda$.
We will need the following two facts about Poisson distributions in the paper.
\begin{lemma}
	Suppose $\{X_i\}_{i = 1}^n$ is a set of independent random variables with $X_i \sim \mathrm{Pois}(\lambda_i)$, then
	$$
	\sum_{i=1}^n X_i \sim \mathrm{Pois}\left(\sum_{i=1}^n \lambda_i\right).
	$$
\end{lemma}
\begin{lemma}\label{lem:tail_pois}
	$$
	\Pr\left[\mathrm{Pois}(\lambda)\ge 1.2\lambda\right] \le e^{-0.01 \lambda}
	$$
	and
	$$
	\Pr\left[\mathrm{Pois}(\lambda) \le 0.8\lambda\right] \le e^{-0.01 \lambda}.
	$$
\end{lemma}

\begin{proof}
	By standard tail inequalities of Poisson distribution (see Theorem 5.4 in \cite{mitzenmacher2005probability}), 
	$$
	\Pr\left[\mathrm{Pois}(\lambda)\ge x \right] \le e^{-\lambda} (e \lambda / x)^x
	$$
	and
	$$
	\Pr\left[\mathrm{Pois}(\lambda)\le x \right] \le e^{-\lambda} (e \lambda / x)^x.
	$$
	
	Thus for any $\lambda > 0$, we have 
	$$
	\Pr\left[\mathrm{Pois}(\lambda)\ge 1.2\lambda \right] \le e^{-\lambda} (e /1.2)^{1.2\lambda} < e^{-0.01\lambda}
	$$
	and
	\[
	\Pr\left[\mathrm{Pois}(\lambda)\le 0.8\lambda \right] \le e^{-\lambda} (e /0.8)^{0.8\lambda} < e^{-0.01\lambda}.\qedhere
	\]
	
\end{proof}

\subsection{Tensor Ranks}

In this paper we are interested in the approximate tensor rank with respect to the $\ell_{\infty}$ norm. For more on approximate tensor rank with respect to other norms and their applications, see~\cite{swz19} and the references therein. Now we introduce some relevant definitions.

\begin{definition}
	We say a tensor $T \in \mathbb{R}^{n_1 \times n_2 \times \ldots \times n_k}$ is {\em simple} if $T = v_1 \otimes v_2 \otimes \ldots \otimes v_k$ where $v_i \in \mathbb{R}^{n_i}$.
\end{definition}
\begin{definition}
	For a tensor $T \in \mathbb{R}^{n_1 \times n_2 \times \ldots \times n_k}$, its $\rank(T)$ is defined to be the smallest integer $r$ such that $T = \sum_{i=1}^r A_r$ and $A_i$ is simple for all $i \in [r]$.
\end{definition}

\begin{definition}
	For a tensor $T \in \mathbb{R}^{n_1 \times n_2 \times \ldots \times n_k}$, the approximate rank of $T$ is defined as follows:
	\[
	\arank_{\varepsilon}(T) = \min \{\rank(S) \mid  \|T - S\|_{\infty} \le \varepsilon \}.
	\]
	Here $\|\cdot\|_{\infty}$ is the entry-wise $\ell_{\infty}$-norm of a tensor.
\end{definition}

\subsection{Quantum Query Complexity}
In this section we recall some previous results on quantum query complexity. 
Here we emphasize the number of qubits used by the algorithms, which will be crucial when simulating them using classical algorithms (see Appendix \ref{apx:simulate}). 
\begin{definition}
	In the $\FormulaEval$ problem, we are given a formula $\mathcal{F}$ with $\{ \textsf{AND}, \textsf{OR},\textsf{NOT} \}$ basis and $O(1)$ fan-in on $n$ variables $x_1, x_2, \ldots, x_n$.
	In each query, the algorithm gets the value of $x_i$, where $i \in [n]$ is determined by the algorithm. 
	The goal is to evaluate the formula. 
\end{definition}
\begin{theorem}[\cite{ambainis2010any}]\label{thm:alg_fe} 
	The $\FormulaEval$ problem can be solved in $O(n^{1/2 + o(1)})$ queries using $O(\polylog(n))$ qubits, with failure probability at most $1/3$.
\end{theorem}

\begin{remark}
	There is an optimal $O(n^{1/2})$ query algorithm for $\FormulaEval$~\cite{Rei11a}. However, that query algorithm doesn't fit in our applications here for two reasons: (1) the algorithm needs $O(n)$ qubits, which is too much for classical simulation; (2) the algorithm is not \emph{computationally efficient} and it takes too much time to compute the corresponding unitary transformation.
\end{remark}

\begin{definition}
	In the $\ED$ problem, we are given $n$ elements $X = (x_1,x_2,...,x_n) \in [m]^n$.
	In each query, the algorithm gets the value of $x_i$, where $i \in [n]$ is determined by the algorithm. 
	The goal is to decide whether there are two distinct indices $i \neq j$ such that $x_i = x_j$.
\end{definition}
\begin{theorem}[\cite{ambainis2007quantum}]\label{thm:alg_ed}
	The $\ED$ problem can be solved in $O(n^{2/3})$ queries using $O(n^{2/3} \log m)$ qubits, with failure probability at most $1/3$.
\end{theorem}


\subsection{Multiparty Quantum Communication Protocols}\label{sec:def_quantum_comm}
In this section, we give our definition of multiparty quantum communication protocols.

Let $\domainX_1, \domainX_2, \ldots, \domainX_k$ be finite sets and $f :  \domainX_1, \domainX_2, \ldots, \domainX_k \to \{0, 1\}$ be a function.
In a $k$-part quantum communication protocols, there are $k$ players $P_1, P_2, \ldots, P_k$, together with a Hilbert space $H = H_{1} \otimes H_2 \otimes \ldots \otimes H_k \otimes \overline{H}$. Here $H_i$ serves as the inner working space for player $P_i$, and $\overline{H}$ is the communication channel between all the players. Each player $P_i$ receives an input $x_i \in \domainX_i$ and the goal is to determine $f(x_1, x_2, \ldots, x_k)$.

Now we give the formal definition of a $k$-party quantum communication protocol. 
\begin{definition}
	A $k$-part quantum communication protocol $\mathcal{P} = \mathcal{P}(x_1, x_2, \ldots, x_k)$ is a sequence of $r$ unitary transforms $\mathcal{P} = (U_1^{p_1}(x_{p_1}), U_2^{p_2}(x_{p_2}), \ldots, U_r^{p_r}(x_{p_r}))$, such that:
	
	\begin{itemize}
		\item $U_i^{p_i}(x_{p_i})$ is a unitary transform {\em acting on} $H_{p_i} \otimes \overline{H_i}$ where $\overline{H_i}$ is a subspace spanned by some qubits of $\overline{H}$\footnote{i.e., $U_i^{p_i}(x_{p_i})$ does not alter qubits other than those in $H_{p_i} \otimes \overline{H_i}$.}. That is, it is the action of $p_i$-th player $P_{p_i}$, who is in charge of the $i$-th turn.
		
		\item The sequence $p_1, p_2, \ldots, p_r$, and $\overline{H_1}, \overline{H_2}, \ldots, \overline{H_r}$ are fixed and do not depend on $x_1, x_2, \ldots, x_k$. In other words, $\overline{H_i}$ corresponds to the qubits in the channel $\overline{H}$ that player $P_{p_i}$ will modify during its action in the $i$-th turn, and all players take actions in a fixed, predefined order.
		
		\item The communication complexity of $\mathcal{P}$ is defined to be $C(\mathcal{P}) = \sum_{i=1}^r \log(\dim(\overline{H_i}))$.
		The space complexity of $P_i$ is defined to be $S_i(\mathcal{P}) = \log(\dim(H_i \otimes \overline{H}))$.
	\end{itemize}
	
\end{definition}

For a protocol $\mathcal{P} = (U_1^{p_1}(x_{p_1}), U_2^{p_2}(x_{p_2}), \ldots, U_r^{p_r}(x_{p_r}))$, we say $\mathcal{P}$ computes $f$ with error $\varepsilon$ if we measure the {\em first} qubit in $\overline{H}$
on the state $U_r^{p_r}(x_{p_r}) \cdot U_{r-1}^{p_{r - 1}}(x_{p_{r - 1}}) \cdot \ldots \cdot U_2^{p_2}(x_{p_2}) \cdot U_1^{p_1}(x_{p_1}) \cdot \ket{0}$, we get $f(x_1, x_2, \ldots, x_k)$ with probability at least $1 - \varepsilon$, for all $x_1 \in \domainX_1, x_2 \in \domainX_2, \ldots, x_k \in \domainX_k$.

\begin{remark}
	We remark that our definition here is more complicated than the usual definition of quantum communication protocols in the literature (see, e.g., \cite{kremer1995quantum}), but nonetheless, it is equivalent to them. We choose to formulate it in such a way because it is easier to describe the classical simulation of quantum communication protocols for low approximate rank decompositions (see Section~\ref{apx:simulate}), and the simulation of quantum query algorithms (see below).
\end{remark}

\subsubsection{Simulating Quantum Query Algorithm in Quantum Communication Protocols}

Quantum communication protocols can be built upon quantum query algorithms (see, e.g., \cite{buhrman1998quantum}).
Here we give an example to show how to simulate a quantum query algorithm for $\FormulaEval$ to construct a quantum communication protocol for the communication problem corresponding to Problem~\ref{pro:count_sym}, under our definition. 

In the corresponding $k$-party communication problem, there are $k$ players, and the $i$-th player $P_i$ is given a vector $u_i \in [r]^{d}$. There are $d$ functions $f_1, f_2, \ldots, f_d$ where each $f_i$ is from $[r]^k$ to $\{0, 1\}$, and a Boolean formula $\mathcal{F} : \{0, 1\}^d \to \{0, 1\}$ of $O(1)$ fan-in. Set $v(i) = f_i(u_{1, i}, u_{2, i}, \ldots, u_{k, i}).$
Their goal is to compute
$
\mathcal{F}(v(1),v(2),\dotsc,v(d)).
$

Now we show how to construct a quantum communication protocol for the above problem.

\begin{example}\label{ex:simluate}
	Assuming the first player runs a quantum query algorithm for $\FormulaEval$.
	For the simulation, we only need to implement the following query gate $O_v$: $\ket{i}\ket{b} \to \ket{i} \ket{b \oplus v(i)}$,
	where $i$ is the index of a variable written in binary form and $v(i)$ is the corresponding input bit to $\mathcal{F}$.
	
	\newcommand{\Hindex}{\overline{H}_{\sf index}}
	\newcommand{\Houtput}{\overline{H}_{\sf output}}
	
	We first specify the channel, $\overline{H}$ is defined as $\Hindex \otimes \Houtput \otimes \overline{H}_1 \otimes \cdots \otimes \overline{H}_k$. $\Hindex$ and $\Houtput$ together simulate the query gate, and $\overline{H}_i$ is the place for player $P_i$ to write her number.
	
	In the beginning, all qubits in $\overline{H}$ are $\ket{0}$. When the first player wants to apply $O_v$ on some qubits in $H_1$, it first swaps the qubits containing $i$ and $b$ in $H_1$ with $\Hindex$ and $\Houtput$ in $\overline{H}$.
	
	Each player $P_j$ in turn reads $i$ in $\Hindex$ and writes the value of $u_{j, i}$ to qubits in $\overline{H}_{j}$. Note that each player can write the value of $u_{j, i}$ to qubits in $\overline{H}_j$ using a unitary transformation since all qubits in $\overline{H}_j$ are $\ket{0}$ at the beginning, by assumption.
	
	Now, given the value of $i$ and $u_{1,i }, u_{2, i}, \ldots, u_{k, i}$, the first player maps $\ket{i}\ket{b}$ to $\ket{i} \ket{b \oplus v(i)}$ via a unitary transformation. Now the gate $O_v$ is implemented, but we still have to clean up the garbages in $\overline{H}_j$'s, and set them back to $\ket{0}$'s. This can be done by applying reverse transforms of all applied unitary transformation, in the reverse order.
	
	The communication complexity of this protocol is $O(Q\cdot k(\log d + \log r))$, where $Q$ is the query complexity of the quantum query algorithm. Also, using the algorithm in Theorem \ref{thm:alg_fe}, the communication complexity of this protocol is $O(n^{1/2+o(1)}\cdot k(\log d + \log r))$.
\end{example}
\section{Approximate Counting Algorithms from Quantum Communication Protocols}

Let $\domainX_1, \domainX_2, \ldots, \domainX_k$ be finite sets and $f :  \domainX_1, \domainX_2, \ldots, \domainX_k \to \{0, 1\}$ be a function.
Let $M_f \in \{0, 1\}^{|\domainX_1| \times |\domainX_2| \times \ldots \times |\domainX_k|}$ denote the Boolean tensor whose $(x_1, x_2, \ldots, x_k)$ entry is $f(x_1, x_2, \ldots, x_k)$.
The following connection between $2$-party quantum communication complexity and approximate rank is first observed in \cite{buhrman2001communication}.
This result can be generalized to the $k$-party case to get the following theorem. 
Full details can be found in Appendix \ref{apx:simulate}.
\begin{theorem}\label{thm:apx_rank}
	Let $\domainX_1, \domainX_2, \ldots, \domainX_k$ be finite sets and $f :  \domainX_1, \domainX_2, \ldots, \domainX_k \to \{0, 1\}$ be a Boolean function.
	Suppose there exists a $k$-party efficient quantum communication protocol $\mathcal{P}$, such that $\mathcal{P}$ gives the correct answer with probability at least $1 - \varepsilon$ on every input,
	then $\arank_{\varepsilon}(M_f) \le 2^{O(C(\mathcal{P}))}$,
	or equivalently, there exist simple tensors $A_1, A_2, \ldots, A_{2^{O(C(\mathcal{P}))}}$ such that
	$$
	\left \|M_f - \sum_{i=1}^{2^{O(C(\mathcal{P}))}} A_i \right\|_{\infty} \le \varepsilon.
	$$
\end{theorem}

In Appendix \ref{apx:simulate} we further show how to use classical deterministic algorithms to simulate quantum communication protocols.
Notice that here the time complexity depends on the space complexity of the quantum communication protocol to use.
\begin{corollary}\label{cor:apx_count}
	Let $\domainX_1, \domainX_2, \ldots, \domainX_k$ be finite sets and $f :  \domainX_1, \domainX_2, \ldots, \domainX_k \to \{0, 1\}$ be a Boolean function.
	Suppose there exists a $k$-party efficient quantum communication protocol $\mathcal{P}$, such that $\mathcal{P}$ gives the correct answer with probability at least $1 - \varepsilon$ on every input, and all the unitary transformation used in the $\mathcal{P}$ can be constructed in polynomial time (with respect to their sizes) by a deterministic classical algorithm.
	Then there exists $k$ deterministic classical algorithms $\alg_{\domainX_1}, \alg_{\domainX_2}, \ldots, \alg_{\domainX_k}$ 
	such that $\alg_{\domainX_i}$runs in $2^{O(C(\mathcal{P}) + S_i(\mathcal{P}))}$ time, 
	receives $x_i \in \domainX_i$ as input and outputs a vector $\alg_{\domainX_i}(x_i) \in \mathbb{R}^{2^{O(C(\mathcal{P}))}}$, 
	and for any $x_1 \in \domainX_1, x_2 \in \domainX_2, \ldots, x_k \in \domainX_k$, 
	$$
	-\varepsilon
	\le 
	\left \langle  \alg_{\domainX_1}(x_1),   \alg_{\domainX_2}(x_2), \ldots,   \alg_{\domainX_k}(x_k) \right \rangle
	- f(x_1, x_2, \ldots, x_k) 
	\le 
	\varepsilon.
	$$
\end{corollary}

Based on Corollary \ref{cor:apx_count}, for any Boolean function $f :  \domainX_1, \domainX_2, \ldots, \domainX_k \to \{0, 1\}$ with an efficient efficient quantum communication protocol, there also exists an efficient approximate counting algorithm for $f$.
\begin{theorem}\label{thm:apx_count}
	Let $\domainX_1, \domainX_2, \ldots, \domainX_k$ be finite sets and $f :  \domainX_1, \domainX_2, \ldots, \domainX_k \to \{0, 1\}$ be a Boolean function.
	Suppose there exists a $k$-party efficient quantum communication protocol $\mathcal{P}$, such that $\mathcal{P}$ gives the correct answer with probability at least $1 - \varepsilon$ on every input, and all the unitary transformation used in the $\mathcal{P}$ can be constructed in polynomial time (with respect to their sizes) by a deterministic classical algorithm.
	Then there exists a classical deterministic algorithm $\apxc$ that receives $X_1 \subseteq \domainX_1, X_2 \subseteq \domainX_2, \ldots, X_k \subseteq \domainX_k$ as input, and outputs a number $E$ such that
	\begin{align*}
	\left|
	\sum_{x_1 \in X_1, x_2 \in X_2, \ldots, x_k \in X_k} f(x_1, x_2, \ldots, x_k) - E
	\right|
	\le \epsilon \cdot \prod_{i=1}^{k} |X_i| .
	\end{align*}
	Furthermore, $\apxc$ runs in $\sum_{i=1}^k |X_i| \cdot 2^{C(\mathcal{P}) + S_i(\mathcal{P})}$ time.
\end{theorem}
\begin{proof}
	For all $x_i \in \domainX_i$ we first use $\alg_{\domainX_i}$ in Corollary \ref{cor:apx_count} to calculate $\alg_{\domainX_i}(x_i) \in \mathbb{R}^{2^{O(C(\mathcal{P}))}}$, in $\sum_{i=1}^k |X_i| \cdot 2^{C(\mathcal{P}) + S_i(\mathcal{P})}$ time. Then we directly output 
	$$
	\left \langle \sum_{x_1 \in X_1} \alg_{\domainX_1}(x_1),  \sum_{x_2 \in X_2}  \alg_{\domainX_2}(x_2), \ldots,  \sum_{x_k \in X_k}  \alg_{\domainX_k}(x_k) \right \rangle .
	$$
	The correctness simply follows from the fact that for all $(x_1,x_2,\dotsc,x_k) \in \prod_i \domainX_i$,
	$$
	-\varepsilon
	\le 
	\left \langle  \alg_{\domainX_1}(x_1),   \alg_{\domainX_2}(x_2), \ldots,   \alg_{\domainX_k}(x_k) \right \rangle 
	- f(x_1, x_2, \ldots, x_k) 
	\le 
	\varepsilon.$$
\end{proof}
\begin{remark}
	The algorithm described above is actually a sketching algorithm.
	We may define the sketch for $X_i$ as  $\sketch_i(X_i) = \sum_{x_i \in X_i} \alg_{\domainX_i}(x_i) \in \mathbb{R}^{2^{O(C(\mathcal{P}))}}$ and the number $E$ can be computed from these $\sketch_i$'s.
	This sketching algorithm satisfies a nice additive property, i.e., the sketch of $A \sqcup B$ (union as a multi-set) is simply $\sketch_i(A) + \sketch_i(B)$.
\end{remark}
Now we give approximate counting algorithms for concrete problems, using Theorem \ref{thm:apx_count}.
\subsection{Counting the $k$-Tuples of Orthogonal Vectors}\label{sec:count_k_ov}
The goal of this section is to prove the following theorem.

\begin{reminder}{Theorem \ref{thm:count_k_ov}}
	For any integers $k$, $d$ and any constant $\epsilon > 0$, $\countkOV_{n,d}$ can be approximated deterministically with additive error $\epsilon \cdot n^k$ in $n \cdot 2^{O(k\sqrt{d})}$ time. In particular, it runs in $n^{1 + o(1)}$ time $k$ is a constant and $d = o(\log^2 n)$.
\end{reminder} 

We first consider quantum communication protocols for the following function $f$.
\begin{definition}
	Let $\domainX_1 = \domainX_2 = \ldots = \domainX_k = \{0, 1\}^d$ and 
	$$
	f(x_1, x_2, \ldots, x_k) = 
	\begin{cases}
	1 & \text{if $\langle x_1, x_2, \ldots, x_k \rangle = 0$}\\
	0 & \text{otherwise}
	\end{cases}.
	$$
\end{definition}
The corresponding communication problem can be solved using the quantum communication protocol in~\cite{AA05} with communication complexity $O(k\sqrt{d})$ and space complexity $O(\polylog(d))$, with constant failure probability.
If we use the algorithm in Theorem \ref{thm:apx_count}, together with the efficient quantum communication protocol mentioned above, we can then deterministically count the number of $k$-tuples of orthogonal vectors, in time $n \cdot 2^{O(k \cdot \sqrt{d})}$ time, with an additive $\epsilon \cdot n^k$ error. 

\subsection{Counting the Pairs of Orthogonal Sparse Vectors}\label{sec:count_osv}
The goal of this section is to prove the following theorem.

\begin{reminder}{Theorem \ref{thm:count_sparse_ov}}
	For integers $n,m,d$ and any constant $\epsilon > 0$, $\sparsecountOV_{n,m,d}$ can be approximated deterministically with additive error $\epsilon \cdot n^2$ in 
	\[
	n \cdot 2^{O(d^{2/3} \log(m))}
	\]
	time. In particular, when $m = \poly(d)$ and $d = o\left( \left(\frac{\log n}{\log 
		\log n} \right)^{1.5} \right)$, it runs in $n^{1 + o(1)}$ time.
\end{reminder}

Again we consider quantum communication protocols for the following function $f$.
\begin{definition}
	Let $\domainX = \domainY =   \{0, 1\}^m_{\le d}$ and 
	$$
	f(x, y) = 
	\begin{cases}
	1 & \text{if $\langle x, y\rangle = 0$}\\
	0 & \text{otherwise}
	\end{cases}.
	$$
\end{definition}

The corresponding communication problem can be solved with communication complexity $O(d^{2/3} \log m)$, by simulating the quantum query algorithm in Theorem \ref{thm:alg_ed} for~\ED{}.
Too see the connection, let $S = \{i \mid x_i = 1\}$ and $T = \{i \mid y_i = 1\}$.
We will have $f(x, y) = 1$ if and only if  all elements in $S \sqcup T$ (union as a multi-set) are distinct.
Now, using the algorithm in Theorem \ref{thm:apx_count}, together with the efficient quantum communication protocol mentioned above, we can deterministically count the number of orthogonal pairs in $S$ and $T$, in $n \cdot 2^{O(d^{2/3} \log(m))}$ time, with an additive $\varepsilon \cdot n^k$ error.

\subsection{Counting Solutions to $\mathsf{Formula} \circ \mathsf{SYM}$ Circuits}
The goal of this section is to solve the following problem.

\begin{reminder}{Problem \ref{pro:count_sym}}
	Given $k$ sets of $n$ vectors $S_1, S_2, \ldots, S_k \subseteq \{0,\dotsc,r\}^{d}$ and $d$ functions $f_1, f_2, \ldots, f_d$ where each $f_i$ is from $\{0,\dotsc,r\}^k$ to $\{0, 1\}$,
	and a Boolean formula $\mathcal{F} : \{0, 1\}^d \to \{0, 1\}$ of $O(1)$ fan-in.
	Count the number of $k$-tuples $u_1 \in S_1, u_2 \in S_2, \ldots, u_k \in S_k$ such that
	$$
	\mathcal{F}(f_1(u_{1, 1}, u_{2, 1}, \ldots, u_{k, 1}),  f_2(u_{1, 2}, u_{2, 2}, \ldots, u_{k, 2}), \ldots, f_d(u_{1, d}, u_{2, d}, \ldots, u_{k, d})) = 1.
	$$
\end{reminder}
\vspace{-1cm}

\begin{reminder}{Theorem \ref{thm:count_sym_split}}
	For any constant $\epsilon > 0$, the above problem can be solved deterministically in $n \cdot 2^{O(d^{1/2+o(1)} \cdot k (\log d + \log r) )}$ time, within $\varepsilon \cdot n^k$ additive error. 
\end{reminder}

The corresponding $k$-party communication problem can be solved by a quantum communication protocol with communication complexity $O(d^{1/2+o(1)} \cdot k (\log d + \log r))$, 
by simulating the quantum query algorithm for \textsf{Formula}-Evaluation in Theorem \ref{thm:alg_fe}.
For details see Example \ref{ex:simluate}.
By our framework, this implies an approximate counting algorithm to the problem mentioned above in time $n \cdot  2^{O(d^{1/2+o(1)} \cdot k (\log d + \log r) )}$, with an additive $\varepsilon \cdot n^k$ error. 

Here we mention one application to the approximate counting algorithm above. 

\begin{reminder}{Theorem~\ref{thm:count_formula_sym}}
	For any constant $\epsilon > 0$, the number of solutions to a $\textsf{Formula} \circ \SYM$ circuit of size $m$ can be approximated deterministically within $\epsilon \cdot 2^n$ additive error in
	\[
	2^{O(n^{1/2} m^{1/4 + o(1)} \sqrt{\log n + \log m})}
	\]
	time. In particular, when $m = n^{2-\delta}$ for some $\delta > 0$, the running time is $2^{o(n)}$.
\end{reminder}

\begin{proofof}{Theorem~\ref{thm:count_formula_sym}}
	Consider a $\mathsf{Formula} \circ \mathsf{SYM}$ circuit $\mathcal{C} : \{0, 1\}^n \to \{0, 1\}$ with $m$ symmetric gates $X_1, X_2, \ldots, X_m$
	and a Boolean formula $\mathcal{F}$ of $O(1)$ fan-in.
	Here we slightly abuse of notation by regarding $X_i$ as a function that maps the number of inputs bits with value one to an output in $\{0, 1\}$.
	We can approximately count the number of solutions to $\mathcal{C}$ as follows.
	
	We split the $n$ inputs bits into $s$ groups, each with $n / s$ input bits. 
	Then for each group, we enumerate all the $2^{n / s}$ possible assignments to the $n / s$ input bits.
	We create a vector in $\{0,\dotsc,n / s\}^{m}$ for each possible assignment, where the $i$-th entry is simply the number of ones in the assignment which is an input bit to the $i$-th symmetric gate $X_i$.
	Now, the number of solutions to the circuit $\mathcal{C}$, is simply the same as Problem \ref{pro:count_sym}, by setting 
	$$
	f_i(u_{1, i}, u_{2, i}, \ldots, u_{k, i}) = X_i(u_{1, i} + u_{2, i} + \ldots + u_{k, i}).
	$$
	The total time complexity would be $2^{n / s} \cdot 2^{O\left( m^{1/2+o(1)} \cdot s (\log m +  \log (n / s))\right)}$, with an additive $\varepsilon \cdot 2^n$ error.
	Setting $s = \frac{n^{1/2}}{m^{1/4+o(1)} \sqrt{\log n+  \log m}}$, the final time complexity would be $2^{O(n^{1/2} m ^{1/4+o(1)} \sqrt{\log n + \log m})}$.
\end{proofof}
\section{Algorithms from Arthur-Merlin Communication Protocols}

In this section, we prove our algorithmic applications of $\AMcc$ protocols. We first show faster $\AMcc$ protocols for $F$ imply faster $F\SATPAIR$ algorithms.


\begin{reminder}{Theorem~\ref{thm:alg_amcc}}
	Let $F : \domainX \times \domainY \to \{0,1,\perp\}$ be a partial function. Suppose there is a computationally efficient $\AMcc$ protocol for $F$ with communication complexity $T$ and error $\epsilon$. Then for $n$ such that $2^{T} \le (\sqrt{\varepsilon}n)^{0.1}$, there is an $O\left( \epsilon n^2 \cdot \polylog(n)  + n \cdot 2^T\right)$ time randomized algorithm for $F\SATPAIR_n$.
\end{reminder}
\begin{proof}
	We first assume $n <\frac{1}{10 \sqrt{\varepsilon}}$. 
	After drawing a random challenge, for each element $x \in \domainX$ and $y \in \domainY$ we construct a Boolean vector $\alg_{\domainX}(x)$ and $\alg_{\domainY}(y)$ of length $2^T$, where each the $i$-th entry indicates whether Alice (Bob) accepts when receiving the proof $i$ from Merlin. Here we regard $i$ as a Boolean string of length $T$ via a natural bijection between $[2^T]$ and $\{0,1\}^T$.
	
	According to the guarantee of an $\AMcc$ protocol, for each $x \in \domainX$ and $y \in \domainY$, when $F(x, y) = 1$, with probability at least $1 - \varepsilon$ over the random challenge, we have $\langle \alg_{\domainX}(x), \alg_{\domainY}(y) \rangle > 0$, and when $F(a, b) = 0$ we have $\langle \alg_{\domainX}(x), \alg_{\domainY}(y) \rangle > 0$ with probability at most $\varepsilon$ over the random challenge.
	
	By a union bound on all pairs of elements in $A$ and $B$, we have with probability at least $0.99$, for all $a \in A$ and $b \in B$, $\langle \alg_A(a), \alg_B(b) \rangle > 0$ if and only if $F(a, b) = 1$.
	Consequently, with probability at least $0.99$,
	$$
	\left \langle \sum_{a \in A} \alg_A(a), \sum_{b \in B} \alg_B(b) \right \rangle > 0
	$$
	if and only if there exist $a \in A$ and $b \in B$ such that $F(a, b) = 1$.
	
	For general $n = |A| = |B|$, we first split $A$ and $B$ into $O(\sqrt{\varepsilon} n)$ groups, each with at most $\frac{1}{10 \sqrt{\varepsilon}}$ elements. 
	I.e., we assume $A = \bigcup_{i = 1}^{g} A_i$ and $B = \bigcup_{i=1}^{g} B_i$ such that $g = O(\sqrt{\varepsilon} n)$ and $|A_i|, |B_i| \le \frac{1}{10 \sqrt{\varepsilon}}$.
	For each $i, j \in [g]$, we use the algorithm mentioned above to calculate two vectors $\sum_{a \in A_i} \alg_A(a)$ and $\sum_{b \in B_j} \alg_B(b)$.
	We write $\mathcal{M}_{A} \in \mathbb{R}^{2^T \times g}$ to denote the matrix 
	$$\left[\sum_{a \in A_1} \alg_A(a), \sum_{a \in A_2} \alg_A(a), \cdots, \sum_{a \in A_g} \alg_A(a)\right]$$
	and $\mathcal{M}_{B} \in \mathbb{R}^{2^T \times g}$ to denote the matrix
	$$\left[\sum_{b \in B_1} \alg_B(b), \sum_{b \in B_2} \alg_B(b), \cdots, \sum_{b \in B_g} \alg_B(b)\right].$$
	Since $2^{T} \le (\sqrt{\varepsilon}n)^{0.1} \le O(g^{0.1})$, we can use the rectangular matrix multiplication algorithm in Theorem \ref{theo:fast-matrix-mult-polylog} to calculate $\mathcal{M}_A^T \mathcal{M}_B$ in $O(g^2 \cdot \polylog(g)) = O(\varepsilon n^2 \polylog(n))$ time. 
	We repeat this procedure for $O(\log n)$ times.
	For any $i, j \in [g]$, by standard concentration bounds, with probability at least $1 - \poly(n)$, there exist $a \in A_i$ and $b \in B_j$ such that $F(a, b) = 1$ if and only if the majority of the $O(\log n)$ repetitions satisfies $(\mathcal{M}_A^T \mathcal{M}_B)_{i, j} > 0$.
	Applying union bound again over all $i, j \in [g]$, we can now solve $F\SATPAIR_n$ by checking whether there exist $i$ and $j$ such that the majority of the $O(\log n)$ repetitions satisfies $(\mathcal{M}_A^T \mathcal{M}_B)_{i, j} > 0$.
	The overall algorithm runs in $O(\varepsilon n^2 \cdot \polylog(n))$ time and succeeds with high probability, as stated.
\end{proof}

\subsection{A New Algorithm for Approximate $\MaxIP$}
\label{sec:amcc-max-ip}

The first application of Theorem~\ref{thm:alg_amcc} is to use the Goldwasser-Sisper $\AM$ protocol~\cite{GS89} for approximating set size to obtain a new algorithm for approximating $\MaxIP$.

We first need the following adaption of~\cite{GS89}, which has a better dependence on $\epsilon$.

\begin{reminder}{Lemma~\ref{lm:AMcc-MGapIP}}
	There is an $\AMcc$ protocol for $\MGapIP_d$ with error $\epsilon$ and communication complexity
	\[
	\log \binom{d}{\le O(\log \epsilon^{-1})}.
	\]
\end{reminder}
\begin{proof}
	Recall that $x,y \in \{0,1\}^d$ are the inputs hold by Alice and Bob respectively.
	
	Let $X = \{i \mid x_i = 1\}$ and $Y = \{i \mid y_i = 1\}$.
	The problem is equivalent to determine whether $|X \cap Y| \ge 2 \tau$ or $|X \cap Y| \le \tau$.
	Here we give an $\AMcc$ communication protocol with error $\epsilon$ and communication complexity
	$
	\log \binom{d}{\le O(\log \epsilon^{-1})}.
	$
	
	In the communication protocol, Alice and Bob first generate i.i.d. random variables $p_i \sim \mathrm{Pois}(k / \tau)$ for each $i \in [d]$, for a parameter $k = \Theta(\log(1 / \varepsilon))$ to be determined later. 
	When $|X \cap Y| \ge 2 \tau$, Merlin finds an arbitrary set $S \subseteq X \cap Y$ of size $O(k)$ such that $\sum_{i \in S} p_i \ge 1.6k$, and then sends it to Alice and Bob.
	Upon receiving $S$, Alice (Bob) decides to accept or reject by checking whether $S \subseteq X$ ($S \subseteq Y$) and $\sum_{i \in S} p_i \ge 1.6k$. 
	The communication complexity of this protocol is upper bounded by $\log \binom{d}{\le O(\log \epsilon^{-1})}$ since $|S| \le 1.6k = O(\log (1 / \varepsilon))$.
	
	Now we prove the correctness by considering the following two cases.
	\begin{description}
		
		\item[Case 1: $|X\cap Y| \ge 2 \tau$. ]
		For this case, we have
		$$\sum_{i \in X \cap Y} p_i \sim \mathrm{Pois}(|X \cap Y| \cdot k / \tau) \succeq \mathrm{Pois}(2k).$$
		Thus by Lemma \ref{lem:tail_pois}, with probability at least $1 - e^{\Omega(k)}$,  $\sum_{i \in X \cap Y} p_i \ge 1.6k$.
		Since for each $p_i > 0$ we must have $p_i \ge 1$,  with probability at least $1 - e^{\Omega(k)}$, there exists a set $S \subseteq X \cap Y$ of size $O(k)$ such that $\sum_{i \in S} p_i \ge 1.6k$.
		
		\item[Case 2: $|X\cap Y| \le \tau$. ] 
		For this case, we have
		$$\sum_{i \in X \cap Y} p_i \sim \mathrm{Pois}(|X \cap Y| \cdot k / \tau) \preceq \mathrm{Pois}(k).$$
		Thus by Lemma \ref{lem:tail_pois}, with probability at least $1 - e^{\Omega(k)}$,  $\sum_{i \in X \cap Y} p_i \le 1.2k$.
		When both Alice and Bob accept, it must be the case that $S \subseteq X \cap Y$ and $\sum_{i \in S} p_i \ge 1.6k$.
		However when $|X\cap Y| \le \tau$, with probability at least $1 - e^{\Omega(k)}$,  $\sum_{i \in X \cap Y} p_i \le 1.2k$.
		Thus there is no $S$ such that both Alice and Bob accept, with probability at least $1 - e^{-\Omega(k)}$.
		
	\end{description}
	The lemma follows by setting $k$ to be a large enough multiple of $\log(1 / \varepsilon)$.
	
\end{proof}

By Theorem~\ref{thm:alg_amcc} and the above lemma, Corollary~\ref{cor:algo-apx-max-ip} follows from a binary search over $\tau$.

\begin{reminder}{Corollary~\ref{cor:algo-apx-max-ip}}
	There is an algorithm for computing a $2$-approximation to $\MaxIP_{n,c \log n}$, which runs in $n^{2 - 1/O(\log c)}$ time.
\end{reminder}

\subsection{Consequence of Fast $\AMcc$ Protocols for $\LCS$ and $\Edit$}

Next we discuss the consequences of $\LCS$ and $\Edit$ having efficient $\AMcc$ protocols. We first introduce some classical notations about the communication complexity classes (see~\cite{BFS86,GPW18}). We say a function family $F = \{F_d : \{0,1\}^{d} \times \{0,1\}^d \to \{0,1,\bot\} \}_{d \in \mathbb{N}}$ is in $\AMcc$ if $\AMcc(F_d) = \polylog(d)$ (we use $\AMcc(F_d)$ to denote the $\AMcc$ communication complexity for $F_d$ with error $1/3$). 

We also say $F$ is $\AMcc_{\sf eff}$ if for all $d \in \mathbb{N}$, $F_d$ admits a computationally efficient $\AMcc$ protocol with error $1/3$ and complexity $\polylog(d)$.

Now we prove the consequence of a function family $F \in \EAMcc$. 

\begin{corollary}[Consequence of $F \in \EAMcc$]\label{thm:alg_eff_am}
	Let $F = \{F_d : \{0,1\}^{d} \times \{0,1\}^d \to \{0,1,\bot\} \}_{d \in \mathbb{N}}$ be a partial function family. If $F \in \EAMcc$, then there is an $n^{2} / 2^{\log^{1-\delta} n}$ time algorithm for $F_{\polylog(n)}\SATPAIR_n$, for any constant $\delta > 0$.
\end{corollary}
\begin{proof}
	By standard repetition arguments, there exists an $\AMcc$ communication protocol with communication complexity $\polylog(d) \log(1 / \varepsilon)$ and failure probability $1 - \varepsilon$.
	In order to invoke Theorem \ref{thm:alg_amcc} we need to make sure 
	$$
	2^{\polylog(d) \log (1 / \varepsilon)} = 2^{\polyloglog(n) \log (1 / \varepsilon)} < n^{0.1},
	$$
	and thus we can set $\varepsilon = 2^{-\log^{1-\delta/2} n}$.
	For this choice of $\varepsilon$ we will then get an $n^{2} / 2^{\log^{1-\delta/2} n} \cdot \polylog (n) \le n^{2} / 2^{\log^{1-\delta} n}$ time algorithm for $F_{\polylog(n)}\SATPAIR_n$, which completes the proof.
\end{proof}

Recall that in $\DistLCS_{d}$ ($\DistEdit_d$), Alice and Bob hold strings $x,y \in \{0,1\}^d$ respectively, and are given an integer $\tau$. Their goal is to decide whether $\LCS(x,y) \ge \tau$ ($\Edit(x,y) \ge \tau$). Now we are ready to prove Theorem~\ref{thm:consequence-LCS}.

\begin{reminder}{Theorem~\ref{thm:consequence-LCS}}
	If $\DistLCS_d$ admits computationally efficient $\AMcc$ protocols with complexity $\polylog(d)$, then Formula-$\SAT$ of polynomial-size formulas admits an $2^{n - n^{1 - \delta}}$ time algorithm for any constant $\delta > 0$. The same holds for $\DistEdit$ in place of $\DistLCS$.
\end{reminder}

We will only discuss $\DistLCS$ here, the proof for $\DistEdit$ follows exactly the same. We first introduce the reduction from \cite{AHVW16} (see also \cite{AB18}).
\begin{theorem}[Implicit in \cite{AHVW16}]\label{thm:lcs_red}
	For a given formula $\mathcal{F}$ with $n$ input variables and size $s$, let $a \in \{0, 1\}^{n / 2}$ be an assignment to first $n / 2$ variables in $\mathcal{F}$ and $b \in \{0, 1\}^{n / 2}$ be an assignment to last $n / 2$ variables in $\mathcal{F}$.
	There exists an algorithm $\alg$ which outputs $G(a) \in \{0, 1\}^{\poly(s)}$ and $\overline{G}(b) \in \{0, 1\}^{\poly(s)}$ such that for a fixed integer $Y$ ($Y$ depends on $\mathcal{F}$),
	\begin{itemize}
		\item $\LCS(G(a), \overline{G}(b)) = Y$ if $a \odot b$ is a satisfying assignment to $\mathcal{F}$;
		\item $\LCS(G(a), \overline{G}(b)) \le Y - 1$ if $a \odot b$ is not a satisfying assignment to $\mathcal{F}$.
	\end{itemize}
\end{theorem}

\begin{proofof}{Theorem~\ref{thm:consequence-LCS}}
	For a given formula $\mathcal{F}$ of size $s = \poly(n)$, we first enumerate all $2^{n / 2}$ possible assignments to first $n / 2$ variables in $\mathcal{F}$ and all possible assignments to last $n / 2$ variables in $\mathcal{F}$.
	For each $a \in \{0, 1\}^{n / 2}$ corresponding to an assignment to first $n / 2$ variables in $\mathcal{F}$ and $b \in \{0, 1\}^{n / 2}$ corresponding to an assignment to last $n / 2$ variables in $\mathcal{F}$, we calculate $G(a)$ and $\overline{G}(b)$ using Theorem \ref{thm:lcs_red}. Note that all $G(a)$'s and $\overline{G}(b)$'s have length $\poly(s) = \poly(n)$. 
	
	Now suppose $\DistLCS \in \EAMcc$ for $\tau = Y$. Applying Corollary \ref{thm:alg_eff_am} with all possible $G(a)$'s and $\overline{G}(b)$'s, we can solve Formula-$\SAT$ in $2^{n -n^{1 - \delta}}$ time for any constant $\delta > 0$.
\end{proofof}

	\section*{Open Problems and Future Directions}

Here we list a few interesting open problems stemming from this work.

\begin{itemize}
	\item In this work, we applied $\BQP^\cc$ and $\AM^\cc$ protocols for the algorithmic purpose. Can we find algorithmic applications of other communication protocols? 
	
	\item Or less ambitiously, can we find more interesting algorithmic applications with other known $\BQP^\cc$ or $\AMcc$ protocols? (this could even be a motivation to find \emph{new} $\BQP^\cc$ or $\AMcc$ protocols!)
	
	\item Our additive approximation algorithm for $\countOV$ runs in near-linear time when $d = o(\log^2 n)$. Is it possible to design a near-linear time algorithm for $d = n^{o(1)}$ dimensions? Note that by a simple Chernoff bound, there is a deterministic $n^{1+o(1)}$ time algorithm with $n^{1+o(1)}$ advice for additive approximations to $\countOV$. So there is a hope to construct such an algorithm. 
	
	\item Our results show that under the hypothesis of~\cite{AB18}, $\DistLCS$ and $\DistEdit$ do not admit computationally efficient $\AMcc$ or $\PHcc$ protocols. Can one prove that \emph{unconditionally}?
	
	\item Is it possible to connect these algorithms from $\AMcc$ or $\PHcc$ protocols to R. Williams' algorithmic approach to circuit lower bounds~\cite{Wil13,Wil14ACC,Murray-Williams17}? In particular, can one show \emph{unconditionally} that, there is a function $f$ in $\textsf{NEXP}$ (or even $\textsf{NTIME}[2^{\polylog(n)}]$), which doesn't admit $\polylog(n)$ complexity $\AMcc$ or $\PHcc$ protocols?
\end{itemize}
	
	\section*{Acknowledgments}
	
	The first author is grateful to Josh Alman, Chi-Ning Chou, Mika G\"{o}\"{o}s, and Ryan Williams for helpful discussions during this work. We are grateful to anonymous reviewers for many helpful and inspiring comments on this paper. In particular, we thank one anonymous reviewer for pointing out that Theorem~\ref{thm:count_sparse_ov} can be improved using the polynomial method.
	
	\bibliographystyle{alpha}
	\bibliography{team}

	\appendix
	\section{Probabilistic Rank and $\OV$ Algorithms in~\cite{AWY15}}
\label{app:prob-low-rank-OV}

In this appendix we explain~\cite{AW17}'s observation with the $\OV$ algorithm in~\cite{AWY15} as an example.~\cite{AWY15} derived an $n^{2 - 1/O(\log c)}$ time algorithm for $\OV_{n,c\log n}$ with the classical polynomial methods~\cite{Raz87,Smo87} (i.e., the probabilistic polynomials for $\AND$ and $\OR$). Here we demonstrate that their results ultimately rely on the fact that the $\SETDISJ$ matrix has a low probabilistic rank.

\newcommand{\distrM}{\mathcal{M}}

For the rest of the section we will always work with $\mathbb{F}_2$-matrices and vectors. A probabilistic matrix $\distrM$ is a distribution of matrices over $\mathbb{F}_2^{n \times n}$. We say a probabilistic matrix $\distrM$ computes a matrix $A \in \mathbb{F}_2^{n \times n}$ with error $\epsilon$, if for every entry $(i,j) \in [n] \times [n]$,
\[
\Pr_{M \sim \distrM} [ A_{i,j} = M_{i,j}] \ge 1 - \epsilon.
\]

We say a probabilistic matrix $\distrM$ has rank $r$, if the maximum rank of matrices from $\distrM$ is $r$. We define the $\epsilon$-probabilistic rank of a matrix $A$ to be the minimum rank of a probabilistic matrix $\distrM$ computing $A$ with error at most $\epsilon$.

\newcommand{\DISJ}{\textsf{DISJ}}

\newcommand{\MDISJ}{M^{\textsf{DISJ}}}

Consider the following $\SETDISJ$ matrix $\MDISJ \in \mathbb{F}_2^{2^d \times 2^d}$. We use subsets of $[d]$ to index rows and columns of $\MDISJ$, and for subsets $S,T \subseteq [d]$, $\MDISJ_{S,T} = \DISJ(S,T)$ ($\DISJ(S,T)$ is the indicator function that whether $S \cap T = \emptyset$).

The following fact is implicit in~\cite{AWY15}:

\begin{proposition}[\cite{AWY15}]\label{prop:AWY-prob-rank}
	The $\epsilon$-probabilistic rank of $\MDISJ$ is $r \le \binom{d}{\le O(\log \epsilon^{-1})}$. Moreover, let $\mathcal{M}$ be the corresponding probabilistic matrix and $M \sim \mathcal{M}$, there are mappings $\phi_X^M,\phi_Y^M : [d] \to \mathbb{F}_2^{r}$ which can be computed in $\poly(r)$ time, such that $ \phi_X^M(S) \cdot \phi_Y^M(T) = M_{S,T}$ for all $S,T \subseteq [d]$.
\end{proposition}

Now, we derive the $\OV$ algorithm in~\cite{AWY15} only using the above proposition.

\begin{theorem}[\cite{AWY15}]\label{theo:AWY15-revisit}
	There is an $n^{2 - 1 / O(\log c)}$ time algorithm for $\OV_{n,c\log n}$.
\end{theorem}
\begin{proof}
	Our presentation here will be different from~\cite{AWY15} for simplicity. Let $d = c\log n$, and $A,B \subseteq \{0,1\}^{d}$ with $|A| = |B| = n$ be the given $\OV$ instance. We say $\OV(A,B) = 1$ if there is an orthogonal pair in $A \times B$, and $\OV(A,B) = 0$ otherwise. 
	
	We first set $\log \epsilon^{-1} = \Theta(\log n / \log c)$ so that after applying Proposition~\ref{prop:AWY-prob-rank}, $r \le n^{0.1}$. 
	
	Let $m = \sqrt{\epsilon^{-1}} / 10$, and assume $m$ divides $n$ for simplicity. We partition $A$ and $B$ into $g = n/m$ groups of vectors, each of size $m$. Let them be $A_1,A_2,\dotsc,A_{g}$ and $B_1,B_2,\dotsc,B_g$ correspondingly.
	
	Let $\mathcal{M}$ be the $\epsilon$-error probabilistic matrix for $\MDISJ$. And $M \sim \mathcal{M}$ be a sample from it. We also draw two random vectors $u,v \in \mathbb{F}_2^{m}$.
	
	Fix two groups $A_i$ and $B_j$, consider the following quantity
	\begin{equation}\label{eq:test}
	\sum_{k=1}^{m} \sum_{\ell=1}^{m} \DISJ( (A_i)_k, (B_j)_\ell ) ) \cdot u_k \cdot v_\ell.
	\end{equation}
	
	Note that this is just $u^{T} W v$, where $W_{k,\ell} = \DISJ( (A_i)_k, (B_j)_\ell ) )$. Since $u$ and $v$ are two random vectors, when $\OV(A_i,B_j) = 0$, \eqref{eq:test} is always zero as $W$ is the all-zero matrix; and otherwise \eqref{eq:test} is $1$ with probability at least $1/4$.
	
	We are going to approximate \eqref{eq:test} by the following
	\begin{align}
	&\sum_{k=1}^{m} \sum_{\ell=1}^{m} M_{(A_i)_k, (B_j)_\ell} \cdot u_k \cdot v_\ell \label{eq:first-with-M} \\
	=&\sum_{k=1}^{m} \sum_{\ell=1}^{m} \{ \phi_X^M((A_i)_k) \cdot \phi_Y^M((B_j)_\ell) \}  \cdot u_k \cdot v_\ell \notag \\
	=&\left\{ \sum_{k=1}^{m} \phi_X^M((A_i)_k) \cdot u_k \right\} \cdot \left\{ \sum_{\ell=1}^{m} \phi_Y^M((B_j)_\ell) \cdot v_\ell
	\right\}. \label{eq:easy-to-eval}
	\end{align}
	
	Note that there are $m^2 = \epsilon^{-1} / 100$ entries of $M$ are considered in \eqref{eq:first-with-M}. So by a union bound and the fact that $\mathcal{M}$ computes $\MDISJ$ with error $\epsilon$. With probability $0.99$, $\eqref{eq:first-with-M}$ and $\eqref{eq:test}$ are equal, over $M$, $u$ and $v$.
	
	Let $U_i = \sum_{k=1}^{m} \phi_X^M((A_i)_k) \cdot u_k$ and $V_j = \sum_{\ell=1}^{m} \phi_Y^M((B_j)_\ell) \cdot v_\ell $. Then by~\eqref{eq:easy-to-eval}, $\eqref{eq:first-with-M}$ equals $U_i \cdot V_j$. Putting everything together, for each $(i,j) \in [g] \times [g]$, we have: (1) when $\OV(A_i,B_j) = 0, \Pr_{M,u,v}[U_i \cdot V_j = 1] \le 0.01$; (2) when $\OV(A_i,B_j) = 1, \Pr_{M,u,v}[U_i \cdot V_j = 1] \ge 0.24$.
	
	For a tuple of $M$, $u$ and $v$, we can compute $U_i \cdot V_j$ for all $i,j$ via a rectangular matrix multiplication between two matrices of size $g \times r$ and $r \times g$. By Theorem~\ref{theo:fast-matrix-mult-polylog}, it can be solved in $g^2 \cdot \polylog(g)$ time. Repeating this for $T = 1000 \log n$ times, for each $i,j$, we record how many times we get $U_i \cdot V_j = 1$ as $T_{i,j}$. The algorithm outputs yes if there are $i,j$ such that $T_{i,j} > T \cdot 0.15$, and no otherwise.
	
	By a simple Chernoff bound, one can show the algorithm solves the $\OV$ instance with high probability, and the running time is $(n/m)^2 \cdot \polylog(n) = \epsilon n^2 \cdot \polylog(n) = n^{2 - 1/O(\log c)}$.
\end{proof}

\begin{remark}\label{rem:prob-rank-to-algo}
	From the above proof, it is easy to see that for any function $F$ whose corresponding communication matrix $M_F$ admits a low probabilistic rank with an efficient decomposition as in Proposition~\ref{prop:AWY-prob-rank}. A fast $F\SATPAIR$ algorithm can be derived similarly.
\end{remark}

	\section{Deterministic Approximate Counting Algorithm for $\textsf{\#OV}$ via Approximate Polynomial}
\label{app:algo-from-approx-poly}
\newcommand{\R}{\mathbb{R}}

Here we show that using the approximate polynomial for $\OR$, one can also derive a deterministic approximate counting algorithm for $\countOV$, but with running time worse than Theorem~\ref{thm:apx-OV}.

\begin{theorem}
	For any $d$ and any $\epsilon > 0$, $\textsf{\#OV}_{n,d}$ can be approximated with additive error $\epsilon \cdot n^2$ in $n \cdot \binom{d}{ \le O(\sqrt{d \log 1/\epsilon})}$ time. In particular, it runs in $n^{1 + o(1)}$ time when $\epsilon$ is a constant and $d = o((\log n/\log \log n)^2)$.
\end{theorem}
\begin{proof}
	By~\cite{BCWZ99,de2008note}, there is a polynomial $P_{\eps} : \{0,1\}^d \to \R$ such that:
	
	\begin{itemize}
		\item $P_{\eps}$ is of degree $D = O\left( \sqrt{d \log 1/\eps} \right)$.
		\item Given $z \in \{0,1\}^{d}$, if $\OR(z) = 0$, then $P_{\eps}(z) \in [1 - \eps, 1] $, otherwise $P_\eps(z) \in [0, \eps]$.
		\item $P_{\eps}$ can be constructed in time polynomial in its description size.
	\end{itemize}

	Let $z_{S} := \prod_{i \in S} z_i$, one can write $P_{\eps}(z) := \sum_{|S| \le D} c_S \cdot z_S$. Let $A,B \subseteq \{0,1\}^{d}$ with $|A| = |B| = n$ be the given $\countOV$ instance, we compute
	\begin{align}
	  E=&\sum_{(x,y) \in A \times B} P_{\eps}(x_1 \cdot y_1,x_2\cdot y_2,\dotsc,x_{d} \cdot y_d) \notag \\
	   =& \sum_{(x,y) \in A \times B} \sum_{|S| \le D} c_S \cdot x_{S} \cdot y_{S}\notag \\
	   =& \sum_{|S| \le D} c_S \cdot \sum_{(x,y) \in A \times B} x_{S} \cdot y_{S}\notag \\
	   =& \sum_{|S| \le D} c_S \cdot \sum_{x \in A} x_{S} \cdot \sum_{y \in B} y_{S} \label{final-line}
	\end{align}
	In above $x_{S}$ and $y_S$ denote $\sum_{i \in S} x_i$ and $\sum_{i \in S} y_i$ respectively.
	
	By the property of $P_{\eps}$, is easy to see that $E$ approximates the number of orthogonal pairs with an additive error $O(\eps \cdot n^2)$. And by \eqref{final-line}, $E$ can be computed in 
	\[
	n \cdot \binom{d}{\le D} = n \cdot \binom{d}{ \le O(\sqrt{d \log 1/\epsilon})}
	\]
	time.
	
	When $\epsilon$ is a constant, the above simplifies to $n \cdot d^{O(\sqrt{d})}$, which is $n^{1 + o(1)}$ when $d = o((\log n /\log\log n)^2)$.
\end{proof}

\begin{remark}\label{rem:also-works}
	We remark that the above algorithm also works for $\sparsecountOV_{n,m,d}$ and $\countkOV_{n,d}$, with running times $n \cdot \binom{m}{\le O\left(\sqrt{d \log 1/\eps}\right)}$ and $n \cdot k \cdot \binom{d}{ \le O(\sqrt{d \log 1/\epsilon})}$, respectively. In particular, it improves the running time in Theorem~\ref{thm:count_sparse_ov}.
\end{remark}


\section{Quantum Communication Protocols and Approximate Rank}\label{apx:simulate}
In this section we explain how to simulate a quantum communication protocol by a deterministic algorithm, and thus prove Theorem \ref{thm:apx_rank} and Corollary \ref{cor:apx_count}.

We first prove the following theorem under our definition of $k$-party quantum communication protocol.
The proof itself follows closely previous proof for $2$-party quantum communication protocols in \cite{buhrman2001communication}.
\begin{theorem} \label{thm:comm_tensor}
For a $k$-party quantum communication protocol.
The final state of $\mathcal{P}$ on input $x_1 \in \domainX_1, x_2 \in \domainX_2, \ldots, x_k \in \domainX_k$, 
can be written as
$$
\sum_{i \in S} a_i^1(x_1) \cdot a_i^2(x_2) \cdot \ldots a_i^k(x_k) \cdot \ket{A_i^1(x_1)}  \ket{A_i^2(x_2)} \ldots \ket{A_i^k(x_k)} \ket{f(i)},
$$
where $a_i^1(x_1) , a_i^2(x_2) , \ldots, a_i^k(x_k) $ are complex numbers and $\ket{A_i^1(x_1)}, \ket{A_i^2(x_2)}, \ldots, \ket{A_i^k(x_k)}$ are unit vectors, 
$
S = \{0, 1\}^{C(\mathcal{P})}
$
is the set of all possible histories of modifications, 
and $f : S \to \overline{H}$ is a function which maps the history of modifications to a state in $\overline{H}$.
\end{theorem}
\begin{proof}
The proof is by induction.
When $r = 0$ the theorem is obvious. 
Suppose after applying $U_{1}^{p_1}, U_{2}^{p_2}, \ldots, U_{r - 1}^{p_{r - 1}}$ the final state is 
$$
\sum_{i \in S'} 
a_i^1(x_1) \cdot a_i^2(x_2) \cdot \ldots a_i^k(x_k) \cdot \ket{A_i^1(x_1)}  \ket{A_i^2(x_2)} \ldots \ket{A_i^k(x_k)} \ket{f(i)},
$$
where
$S' = \{0, 1\}^{\sum_{i = 1}^{r-  1} \log(\dim(\overline{H_i}))}$.
Now we apply $U_r^{p_r}$. 
Since $U_r^{p_r}$ acts on $H_{p_r} \otimes \overline{H_r}$, 
thus every element of the superposition in the previous states splits into at most $2^{\log \dim(\overline{H_r})}$ terms, depending on the state of the qubits in $\overline{H_r}$ after applying $U_r^{p_r}$.
Thus, after applying $U_r^{p_r}$ there will be at most $|S'| \cdot 2^{\log \dim(\overline{H_r})} = 2^{\sum_{i = 1}^{r} \log(\dim(\overline{H_i}))}$ terms in the superposition.
\end{proof}

Now let $S_1$ to be the set that contains all $i \in S$ such that the first qubit in $f(i)$ is $\ket{1}$, and let 
$$\phi(x_1, x_2, \ldots, x_k) = \sum_{i \in S_1} 
 a_i^1(x_1) \cdot a_i^2(x_2) \cdot \ldots a_i^k(x_k) \cdot \ket{A_i^1(x_1)}  \ket{A_i^2(x_2)} \ldots \ket{A_i^k(x_k)} \ket{f(i)}
$$ be the part of the final state that corresponds to
a $1$-output of the protocol.

Now for $i, j \in S_1$, we define
$$
a_{i, j}^p(x_p) = \overline{a_i^p(x_p)} a_{j}^p(x_p) \braket{A_i^p(x_p) | A_j^p(x_p)}.
$$
Thus, the probability of outputting $1$ is 
$$
\braket{\phi(x_1, x_2, \ldots, x_k) | \phi(x_1, x_2, \ldots, x_k)} = \sum_{i, j \in S_1} \prod _{p = 1}^k a_{i, j}^p(x_p).
$$
Thus, we have
$$
\left |\sum_{i, j \in S_1} \prod _{p = 1}^k a_{i, j}^p(x_p) - f(x_1, x_2, \ldots, x_k) \right| \le \varepsilon.
$$

Now we have finished our prove for Theorem \ref{thm:apx_rank}, since the tensor defined by 
$$
A(x_1, x_2, \ldots, x_k) = \prod _{p = 1}^k a_{i, j}^p(x_p)
$$
is a simple tensor, and thus $\rank_{\varepsilon}(M_f) \le |S_1|^2 = 2^{O(C( \mathcal{P}))}$.

Scrutinizing the proof of Thereom \ref{thm:comm_tensor}, for each player $P_p$, for any input $x_p \in \domainX_p$, 
we can calculate $a_{i, j}^p(x_p)$ in $2^{O(S_p(\mathcal{P}))} \cdot \poly(|S|) = 2^{O(S_p(\mathcal{P}) + C(\mathcal{P}))}$ time,
by using a classical deterministic algorithm to simulate the procedure above, as long as all unitary transforms $U_{i}^{p_i}$ can be constructed in polynomial time (with respect to its size).

\section{Conditional Lower Bound for Computational-Efficient $\PHcc$ Protocols}\label{app:PH-protocols}

In this section we prove Theorem~\ref{theo:condtional-lowb-for-PH} (restated below).

\begin{reminder}{Theorem~\ref{theo:condtional-lowb-for-PH}}
	Under the following hypothesis, $\DistLCS_d$ and $\DistEdit_d$ do not admit computationally efficient $\PHcc$ protocols with complexity $\polylog(d)$:
	\begin{itemize}
		\item There is a constant $\delta > 0$ such that Formula-$\SAT$ of polynomial-size formulas requires $2^{n - n^{1-\delta}}$ time.
	\end{itemize}
\end{reminder}

We first recall the definition of a $\PHcc$ protocol from~\cite{BFS86}.\footnote{See also~\cite{GPW18} for a more recent reference.}

\begin{definition}[\cite{BFS86}]\label{defi:PHcc-protocols}
	A $\PH$ communication protocol ($\PHcc$) $\Pi$ for a function $F : \domainX \times \domainY \to \{0,1,\bot\}$ proceeds as follows:
	
	\begin{itemize}
		\item Alice holds input $x \in \domainX$ and Bob holds input $y \in \domainY$.
		\item For a constant $k \in \mathbb{N}$, there are $2 k $ provers $P_1,P_2,\dotsc,P_{2k}$.
		\item For each $i \in [2k]$, the prover $P_{i}$ sends both Alice and Bob a proof $z_i \in \{0,1\}^{m_i}$.
		
		\item We use $A(x,z_1,z_2,\dotsc,z_{2k})$ (resp. $B(y,z_1,z_2,\dotsc,z_{2k})$) to be the indicator function that whether Alice (resp. Bob) accepts the proof sequence $z_1,z_2,\dotsc,z_{2k}$, given the input $x$ (resp. $y$).
		
		\item If $F(x,y) = 1$, then
		\[
		\exists_{z_1 \in \{0,1\}^{m_1}} \forall_{z_2 \in \{0,1\}^{m_2}} \cdots \exists_{z_{2k-1} \in \{0,1\}^{m_{2k-1}}}  \forall_{z_{2k} \in \{0,1\}^{m_{2k}}} \left[ A(x,z_1,z_2,\dotsc,z_{2k}) \wedge B(y,z_1,z_2,\dotsc,z_{2k}) \right].
		\]    
		
		\item If $F(x,y) = 0$, then
		\[
		\forall_{z_1 \in \{0,1\}^{m_1}} \exists_{z_2 \in \{0,1\}^{m_2}} \cdots \forall_{z_{2k-1} \in \{0,1\}^{m_{2k-1}}}  \exists_{z_{2k} \in \{0,1\}^{m_{2k}}} \left[ \neg A(x,z_1,z_2,\dotsc,z_{2k}) \vee \neg B(y,z_1,z_2,\dotsc,z_{2k}) \right].
		\]
	\end{itemize} 
	
	Moreover, we say the protocol is \emph{computationally efficient} if Alice and Bob's decision functions (the functions $A$ and $B$) can be computed in polynomial-time w.r.t. their input lengths. The communication complexity of $\Pi$ is simply the total number of proof bits from all provers, i.e. $\sum_{i=1}^{2k} m_i$.
\end{definition}

\begin{theorem}\label{theo:algo-from-PHcc-protocols}
	Let $F : \domainX \times \domainY \to \{0,1,\perp\}$ be a partial function. Suppose there is a computationally efficient $\PHcc$ protocol for $F$ with communication complexity $T$ and number of provers $2k$. If $\varepsilon > 0$ satisfies
	\[
	\binom{2^T}{\le 10 \cdot (2 \cdot T)^{2k} \cdot \log(1/\epsilon)} \le n^{0.1},
	\]
	then there is an $O(\epsilon \cdot n^2 \cdot \polylog(n))$ time algorithm for $F\SATPAIR_n$.
\end{theorem}
\begin{proof}
	By Remark~\ref{rem:prob-rank-to-algo}, we only need to argue the probabilistic rank of the communication matrix $M_F$ is small, which is already established in~\cite{razborov1989rigid}. In the following, we follow the proof structure of Theorem~\ref{theo:AWY15-revisit}.
	
	Recall that $T = \sum_{i=1}^{2k} m_i$. Let $M = 2^T$ and $z \in \{0,1\}^{M}$. With a natural bijection between $[M]$ and $\{0,1\}^T$, we can use a string $x \in \{0,1\}^T$ to index $z$.
	
	We define the following $\textsf{AC}^0$ function
	\[
	F(z) := \vee_{x_1 \in \{0,1\}^{m_1}} \wedge_{x_2 \in \{0,1\}^{m_2}} \cdots \vee_{x_{2k - 1} \in \{0,1\}^{m_{2k - 1}}} \wedge_{x_{2k} \in \{0,1\}^{m_{2k}}} z_{(x_1 \circ x_2 \circ \dotsc x_{2k})},
	\]
	where $x_1 \circ x_2 \circ \dotsc x_{2k}$ means the concatenation of the $x_i$'s.
	
	By standard polynomial method~\cite{Raz87,Smo87}, there is a
	\[
	D = 10 \cdot (2 \cdot T)^{2k} \cdot \log(1/\epsilon)
	\]
	degree, $\epsilon$-error probabilistic polynomial for the function $F$. Formally, there is an efficiently-sampleable distribution on $D$-degree polynomials $\mathcal{P}_\eps$, such that for all $z \in \{0,1\}^M$,
	\[
	\Pr_{P \sim \mathcal{P}_\eps} [ P(z) = F(z) ] \ge 1-\epsilon.
	\]
	
	In particular, this means the $\epsilon$-probabilistic rank of $M_F$ is smaller than 
	\[
	r \le \binom{M}{\le D},
	\]
	and the corresponding distribution on low-rank matrices has efficiently computable decomposition as specified in Proposition~\ref{prop:AWY-prob-rank}. Then we can proceed exactly as in Theorem~\ref{theo:AWY15-revisit}.
\end{proof}

Finally, we are ready to prove Theorem~\ref{theo:condtional-lowb-for-PH}.

\begin{proofof}{Theorem~\ref{theo:condtional-lowb-for-PH}}
	We proceed similarly as in~Theorem~\ref{thm:consequence-LCS}. Below we only discuss $\DistLCS$, the proof for $\DistEdit$ is exactly the same.
	
	For a given formula $\mathcal{F}$ of size $s = \poly(n)$, we first enumerate all $2^{n / 2}$ possible assignments to first $n / 2$ variables in $\mathcal{F}$ and all possible assignments to last $n / 2$ variables in $\mathcal{F}$. For each $a \in \{0, 1\}^{n / 2}$ corresponding to an assignment to first $n / 2$ variables in $\mathcal{F}$ and $b \in \{0, 1\}^{n / 2}$ corresponding to an assignment to last $n / 2$ variables in $\mathcal{F}$ we calculate $G(a)$ and $\overline{G}(b)$ using Theorem \ref{thm:lcs_red}. We can assume that both $G(a)$ and $\overline{G}(b)$ have length $\ell = \poly(s) = \poly(n)$.
	
	Now suppose $\DistLCS$ with $\tau = Y$ has a computationally efficient $\PHcc$ protocol with $T = \polylog(\ell) = \polylog(n)$ and the number of provers $2k$ ($k$ is a constant).
	We set $\epsilon$ such that
	\[
	\binom{2^T}{\le 10 \cdot (2 \cdot T)^{2k} \cdot \log(1/\epsilon)} \le 2^{n/20}.
	\]
	The above can be satisfied if
	\[
	2^{T \cdot 10 \cdot (2 \cdot T)^{2k} \cdot \log(1/\epsilon)} = 2^{\polylog(n) \cdot \log(1/\epsilon)} \le 2^{n/20}.
	\]
	Finally, we set $\epsilon = 2^{-n^{1-\delta/2}}$ for the $\delta > 0$ in the hypothesis.
	Now we can complete the proof by applying Theorem~\ref{theo:algo-from-PHcc-protocols}.
\end{proofof}

\end{document}